\documentclass[11pt,reqno]{amsart}
\usepackage[foot]{amsaddr}

\usepackage[left=2.5cm, right=2.5cm, top=2cm]{geometry}

\usepackage[english]{babel}

\usepackage{amsmath,amssymb,amsthm,amsfonts,mathrsfs,amsxtra,bbm,comment,amsaddr}

\usepackage{tikz-cd}
\usepackage{tikz}
\usetikzlibrary{arrows.meta} 
\usetikzlibrary{decorations.markings} 

\usepackage[hyperfootnotes=true,
        pdffitwindow=true,
        plainpages=false,
        pdfpagelabels=true,
        pdfpagemode=UseOutlines,
        pdfpagelayout=SinglePage,
        hyperindex,]{hyperref}

\usepackage{appendix}

\numberwithin{equation}{section}

\theoremstyle{plain}
	\newtheorem{theorem}{Theorem}[section]
	\newtheorem{lemma}[theorem]{Lemma}
	\newtheorem{corollary}[theorem]{Corollary }
	\newtheorem{proposition}[theorem]{Proposition}
\theoremstyle{definition}

        \newtheorem{assumption}{Assumption}[section]
        \newtheorem{conjecture}{Conjecture}[section]
\theoremstyle{remark}
	\newtheorem{remark}{Remark}[section]

\newcommand{\N}{\mathbb N}
\newcommand{\Z}{\mathbb Z}
\newcommand{\R}{\mathbb R}
\newcommand{\C}{\mathbb C}

\renewcommand{\O}{\mathrm O}
\newcommand{\U}{\mathrm U}

\newcommand{\E}{\mathbb{E}\,}
\newcommand{\dv}{\mathrm{d}}
\newcommand{\tr}{\mathrm{Tr}\,}

\newcommand{\diag}{\mathrm{diag}}

\newcommand{\sgn}{\mathrm{sgn}}

\newcommand{\ind}{\mathbbm{1}}

\newcommand{\re}{\mathrm{Re}\,}
\newcommand{\im}{\mathrm{Im}\,}

\usepackage{stackengine} 
\newcommand\oast{\stackMath\mathbin{\stackinset{c}{0ex}{c}{0ex}{\ast}{\bigcirc}}}

\begin{document}

\title[Matrix harmonic analysis at high temperature]{Matrix harmonic analysis at high temperature via the Dirichlet process}

\author{Jiyuan Zhang}
\email[JZ]{jiyuanzhang@scut.edu.cn}
\address[JZ]{School of Mathematics, South China University of Technology, China}
\date{\today}

\begin{abstract}
\noindent We investigate harmonic analysis of random matrices of large size with their Dyson indices going simultaneous to zero, that is in the high temperature limit. In this regime, we show that the multivariate Bessel function/Heckman-Opdam hypergeometric function of the empirical spectral distribution converges to the Fourier/Mellin transform of a measure, which and the limiting empirical distribution are intimately related by the Markov-Krein correspondence. The uniqueness, existence and other properties of the Markov-Krein correspondence can be studied using the theory of the Dirichlet process.
\end{abstract}

\maketitle

\section{Introduction and Main results}

\subsection{Overview}

Harmonic analysis studies functions and measures through their spectral representations, most notably via Fourier transforms or characteristic functions. One of the key strengths of harmonic analysis is that all information about a measure (or distribution) is encoded in its characteristic function.

In the matrix setting, for instance, the \textit{Harish-Chandra integral} (or the \textit{spherical integral}~\cite{HC,IZ,Helgason_book}), which generalises classical Fourier transform to non-commutative settings, serves as a fundamental tool in representation theory and harmonic analysis on symmetric spaces. In the perspective of random matrix theory, the spherical integral is explicitly given by
\begin{equation}
    I_N^{(\beta)}(B_N;A_N):=\int_{\U_\beta(N)}\exp\tr (UB_NU^*A_N)\dv U,\quad \beta=1,2.
\end{equation}
where $U_1(N)$ denotes the orthogonal group $\O(N)$; $\U_2(N)$ denotes the unitary group $\U(n)$; $\dv U$ denotes the normalised Haar measure on $U_\beta(N)$ respectively. Taking $B_N=\diag(u,0,\ldots,0)$, we say this is a rank-one version of the spherical integral, denoted by $I_N^{(\beta)}(u;A_N)$. 

A particular interest is the asymptotics of the rank-one spherical integral, when the size of the matrix $N\to+\infty$. A celebrated result~\cite{Benoit_thesis,GZ02,GM05} establishes a profound link between its logarithmic asymptotics and $\mathcal R$-transform from free probability theory, that is, an analogue of the Fourier transform in the non-commutative setting. Explicitly, if the spectral measure (the normalised counting measure of the eigenvalues) of the matrix $A_N$ converges weakly to a deterministic measure $\rho$, then as $N\to \infty$, with certain analytic conditions one has,
\begin{equation}
    \frac{2}{N\beta}\log I_N^{(\beta)}\left(\frac{N\beta}{2}u;A_N\right)\to\int_0^{u}\mathcal R_{\rho}(s)\dv s,
\end{equation}
where $\mathcal R_{\rho}$ is the $\mathcal R$-transform of the measure $\rho$. For recent developments in this direction see~\cite{BG_rectangular, GuionnetHusson2021, HussonKo2022}.
\\

The previous result is shown only for $\beta=1,2$. While classically $\beta$ represents the symmetry of the matrix spaces (with $\beta = 1,2$ corresponding to orthogonal/unitary ensembles), more recent viewpoints treat $\beta$ not merely as a discrete symmetry parameter but as a continuous deformation parameter, reflecting the continuous nature of the inverse temperature of an associated interacting particle system. Two complementary frameworks allow one to extend the theory from the classical Dyson indices to all $\beta>0$.

In the flavour of statistical mechanics, the first framework, called the $\beta$-ensemble, is to identify eigenvalues as a one-dimensional Coulomb gas model subject to a confining potential and pairwise logarithmic repulsion. In this interpretation the parameter $\beta$ plays the role of the strength of repulsion (or inverse temperature), interpolating continuously between the classical symmetry classes and enabling a thermodynamic viewpoint on spectral statistics. A canonical example is the Gaussian $\beta$ ensemble, which has a joint density
\begin{equation}\label{GbetaE}
    p(\dv \lambda_1,\ldots,\dv\lambda_N)\propto \exp\left(\beta\left(\sum_{j<k}\log|\lambda_j-\lambda_k|-\sum_{j=1}^N\frac{ x_j^2}{4}\right)\right).
\end{equation}
We refer the readers to~\cite{AGZ,ForresterBook,DE02,BEY14,BG13} for an overview of the classical $\beta$-ensembles.

Within the same Coulomb–gas interpretation, it is curious to allow the interaction strength $\beta$ itself to vary with the system size. This leads to the study of $N$-dependent $\beta$-ensembles, in which $\beta=\beta_N$ is a prescribed sequence depending on $N$. Of particular interest is the regime $\beta=\beta_N=O(\frac{1}{N})$ known as the \textit{high temperature regime}, where the repulsion becomes asymptotically weak relative to the confining potential. In this regime the repulsion is sufficiently weak that classical log-gas methods must be modified, and many features of the spectrum exhibit a crossover between independent-particle behaviour and strongly correlated Coulomb gas behaviour. Following the pioneering analysis of the high temperature Gaussian $\beta$-ensemble in~\cite{ABG12}, a series of works have investigated high temperature $\beta$-ensembles for more general potentials and related particle systems~\cite{AG21,MP22,DGM24,AB19,FM21,Forrester22,Magaldi22,NT18,NT20,NTT25}. 

Although the $\beta$-ensemble viewpoint motivates why continuous $\beta$ is natural, our main results rely on the second deformation framework---one does not deform the repulsion strength but instead deforms the spherical integral itself. Such framework has been introduced by using theory of Macdonald polynomials (see e.g.~\cite{Macdonald_book,BG15,GM20}). A specific reduction of the Macdonald polynomial is the \textit{multivariate Bessel function} (see e.g.~\cite{Dunkl,Opdam93,dJ93,GK02,Cuenca21,Sun16}), acting as the continuous deformation for the spherical integral, which coincides with the spherical integral for $\beta=1,2$,
\begin{equation}
    \mathcal B_{\vec a}\left(\vec b;N,\frac{\beta}{2}\right)=I_N^{(\beta)}(\diag(\vec a);\diag(\vec b)),
\end{equation}
where $\diag(\vec a)$ denotes the diagonal matrix with eigenvalues $\vec a=(a_1,\ldots,a_N)$. The full definition of the multiplicative Bessel function will be given in the next subsection. Although strictly speaking such a $\U_\beta(N)$ group does not exists for general $\beta$, it helps our understanding to view the sequence $\vec a=(a_1,\ldots,a_N)$ as eigenvalues of a virtual ``$\beta$-matrix'', and thus from a harmonic analysis viewpoint, the multivariate Bessel function can be seen as the characteristic function of such a $\beta$-matrix. 

Again, we would like to consider $\beta$ being an $N$-dependent parameter, i.e. $\beta=\beta_N$ (we will omit the subscript $N$ throughout the rest of the paper). It turns out that there are three regimes of $\beta$:
\begin{alignat*}{2}
    &\text{Classical regime: }&\beta&=o\left(\frac{1}{N}\right),\\
    &\text{High temperature regime: }\quad&\beta&=\frac{2c}{N}+o\left(\frac{1}{N}\right),\text{ for some fixed }c>0,\\
    &\text{Free regime}&\beta&\gg \frac{1}{N},
\end{alignat*}
as $N$ goes to infinity (the second regime gets its name as $\beta$ being called the ``inverse temperature''). The asymptotics the multivariate Bessel function behaves differently in each regime. As mentioned above, when $\beta$ is fixed to be $1$ or $2$ (meaning we are in the free regime), the logarithm of the spherical integral tends to the integral of the $\mathcal R$-transform. In the other two regimes, we would like to understand respective asymptotics of the multivariate Bessel function. It turns out, that in the high-temperature limit, if the normalised counting measure of the vector $\vec a$ has a limiting measure $\rho$, its rank-one multivariate Bessel function tends to the Fourier transform of a different measure $\rho^{(c)}$,
\begin{equation}\label{1.4}
    \mathcal B_{\vec a}\left((u,0,\ldots,0);N,\frac{\beta}{2}\right)\xrightarrow{\text{high-temperature limit}}\int_\R e^{ux}\rho^{(c)}(\dv x)
\end{equation}
where $\rho$ and $\rho^{(c)}$ satisfies the \textit{Markov-Krein correspondence}~\eqref{MK}. Such result has been partially studied in~\cite{MP22,Mergny_thesis,BGCG22}, assuming the existence of all moments. In this work we will extend it to more general limiting measures including the ones with unbounded supports and possibly with heavy tails. We also refer the readers to~\cite{Mergny_thesis,BGCG22,NTT25} for several compelling examples of random matrix ensembles in the high temperature limit.
~\\

The Harish-Chandra integral is intemately related to sums of random matrices. Being the matrix characteristic function, it can be verified that for deterministic matrices $A_N,B_N$ and a Haar matrix $U$, the Harish-Chandra integral admits the following property
\begin{equation}
    \E_U I_N^{(\beta)}(S_N;A_N+UB_NU^*)=I_N^{(\beta)}(S_N;A_N)I_N^{(\beta)}(S_N;B_N).
\end{equation}
In relation to this, there is a multiplicative counterpart of the spherical integral, the Gelfand-Naimark integral~\cite{GN57}, given by
\begin{equation}
    J_N^{(\beta)}(\vec s;A_N):=\int_{\U_\beta(N)}\prod_{k=1}^N\det ([UA_NU^{*}]_{k\times k})^{s_j-s_{j+1}}\,\dv U,\quad \beta=1,2,\quad s_{N+1}=0.
\end{equation}
Here $[X]_{k\times k}$ denotes the $k\times k$ principal minor of the matrix $X$. It enjoys a similar property under a symmetrised matrix multiplication: for deterministic semi-positive definite matrices $A_N,B_N$ and a Haar matrix $U$, one has
\begin{equation}
    \E_U J_N^{(\beta)}(\vec s;A_N^{1/2}UB_NU^*A_N^{1/2})=J_N^{(\beta)}(\vec s;A_N)J_N^{(\beta)}(\vec s;B_N).
\end{equation}
This integral has been used to study products of random matrices; see e.g.~\cite{KZ23}. The large $N$ asymptotics of this integral in the free regime has also been shown (for the special cases $\beta=1,2$), where it recovers the S-transform in free probability theory; see ~\cite{MP20,Husson21}.
~\\

Similar to the relationship between the Harish-Chandra integral and the multivariate Bessel function, when generalising the index $\beta$ to a continuous value, the \textit{Heckham-Opdam hypergeometric function}~\cite{HS94,BG15,Sun16} as the multiplicative counterpart of the multivariate Bessel function comes into play. We remark the relation for $\beta=1,2$ as
\begin{equation}
    \mathcal F_{\log \vec a}\left(\vec s-\frac{\beta}{2}(0,1,\ldots,N-1);N,\frac{\beta}{2}\right)=J_N^{(\beta)}\left(\vec s;\diag(\vec a)\right).
\end{equation}
where $\log \vec a=(\log a_1,\ldots,\log a_N)$. The full definition of the Heckman-Opdam hypergeometric function will be given in the next subsection. 

A parallel phenomenon occurs in the multiplicative setting, governed by the asymptotics of the rank-one Heckman-Opdam hypergeometric function in the high-temperature limit. If the normalised counting measure of the vector $\vec a$ has a limiting measure $\rho$, its rank-one Heckman-Opdam hypergeometric function tends to the Mellin transform of the $\rho^{(c)}$, which appears to be the same measure as ion~\eqref{1.4},
\begin{equation}
    \mathcal F_{\log \vec a}\left((u,0,\ldots,0);N,\frac{\beta}{2}\right)\xrightarrow{\text{high-temperature limit}}\int_\R x^u\rho^{(c)}(\dv x).
\end{equation}
This result, as far as the author is aware of, has not been studied elsewhere. Furthermore, we will also look at the asymptotics of both the multivariate Bessel function and the Heckman-Opdam hypergeometric function in the classical limit. 
~\\

The unforeseen appearance of the Markov-Krein correspondence in this problem leads us to investigate the theory of \textit{Dirichlet process}, a well-studied random process intimately relating with the Dirichlet distribution. It originally arises from the studies of Bayesian non-parametric statistics~\cite{Ferguson73,CRa,CRb,CR90,DK96,LR04,LP09}. The Markov-Krein correspondence itself appears extensively in studies in Markov moment problems~\cite{KM77,Kerov_book}, in Rayleigh measures~\cite{Kerov_paper,Kerov_book,Krein_book,Bufetov13,Faraut,FF16,Fourati11a,Fourati11b,Sodin17}, in asymptotic freeness~\cite{ACG23,MT24}, and in properties of its Fourier transform~\cite{DelloSchiavo19,LR04,Yamamoto84}. The relationship between Dirichlet process and Markov-Krein correspondence is also described in Section~\ref{s2}.
~\\

The purposes of this paper are in three folds:

\begin{itemize}
    \item to recognise the counterpart of the asymptotics of the spherical integrals~\cite{GM05,MP20,Husson21} in the high temperature regime;
    \vspace{.5em}
    \item to emphasise the importance of the Dirichlet process in studies of the Markov-Krein correspondence and the high temperature limits;
    \vspace{.5em}
    \item as a side-product, to provide further properties of the random mean of the Dirichlet process.
\end{itemize}

\subsection{Preliminaries and main results}

\subsubsection{Probability measures and convergence assumptions}

In this paper we consider Borel probability measures on the real line. Their supports can be \emph{non-compact}. 
Let $\rho$ be a Borel probability measure on $\R$ and $c>0$. We assume a logarithmic tail control of the measure $\rho$, that is, $\rho\in\mathcal V$ with
\begin{equation}
    \mathcal V=\left\{\rho:\rho\text{ is a probability measure and }\int_\R\log(1+x^2)\rho(\dv x)<+\infty\right\},\label{V}
\end{equation}
For all such $\rho\in\mathcal V$, there exists an other probability measure $\rho^{(c)}$ such that 
\begin{equation}\label{MK}
    \int\frac{\rho^{(c)}(\dv x)}{(z-x)^c}=\exp\left(-c\int_\R\log(z-s)\rho(\dv s)\right),\quad \text{for all }z\in\C\setminus\R,
\end{equation}
We say that $\rho^{(c)}$ and $\rho$ satisfy a \textit{Markov-Krein correspondence}. The existence and uniqueness of $\rho^{(c)}$ are guaranteed by Theorem~\ref{equiv}.
\\

Let us consider a sequence of discrete measures $\{\rho_N\}$:
\begin{equation}\label{discrete}
    \rho_N(\dv x)=\frac{1}{N}\sum_{j=1}^N\delta_{a_j^{(N)}}(\dv x)
\end{equation}
where $a_1^{(N)}<\ldots<a_{N}^{(N)}$ are points on the real line. They can be considered either deterministic or random; in the latter case the measure $\rho_N$ is a random measure.

The sequence $\{\rho_N\}$ approximates a limiting measure $\rho$ in the large $N$ limit. We want to introduce a specific type of convergence. For $0<\eta<1$, the extended $\eta$-Wasserstein distance is given by
\begin{equation}
    \dv_{W_\eta}(\mu,\nu):=\sup_{h\in\mathcal H_\eta}\left|\int_\R h(x)\left(\mu(\dv x)-\nu(\dv x)\right)\right|
\end{equation}
where $\mathcal H_\eta$ is the class of function $h$ satisfying
\begin{equation}
    |h(x)-h(y)|\le\min\{ |x-y|,|x-y|^\eta\}.
\end{equation}
Then we give the following assumptions.
\begin{assumption}[Convergence for a sequence of deterministic measures]\label{a1}
    There exists $0<\eta<1$ such that the extended $\eta$-Wasserstein distance between $\rho_N$ and $\rho$ is small when $N\to+\infty$.
\end{assumption}
\begin{remark}[Premise of Assumption~\ref{a1}]\ 
    \begin{enumerate}
    \item The convergence of the extended $\eta$-Wasserstein distance between probability measures is a strictly stronger requirement than the weak convergence. The reason why we require such an assumption is, as implied in Lemma~\ref{lem_unif}, such convergence implies a uniform control (away from the support) of the logarithmic integral (the $g$-function). The latter is the essential technical requirements needed for our derivations.
    \vspace{.5em}
    \item Since the extended $\eta$-Wasserstein distance is bounded by the Wasserstein distance, which is bounded above by the Prokhorov metric, we conclude that in the case all $\rho_N$ and $\rho$ are uniformly compactly supported, Assumption~\ref{a1} is equivalent to a weak convergence.
    \vspace{.5em}
    \item In~\cite{GM05} an assumption regarding a small \textit{Dudley distance} between $\rho_N$ and $\rho$ has been considered. In the special case where $\rho_N$ and $\rho$ are uniformly compactly supported, the convergence in Dudley distance is equivalent to the weak convergence. We would like to point out the similarity between our assumption and the one in~\cite{GM05}.
    \vspace{.5em}
    \item This assumption quantifies the convergence of $\rho_N$, which is also is in line with~\cite{CNXYZ22}; the latter considers convergence towards \textit{stable distributions}, an important class of heavy-tailed distributions, which plays the role of the normal distribution in consideration with the central limit theorem. We take the extended $\eta$-Wasserstein distance instead of the usual Wasserstein distance, because the convergence according to the latter does not conclude a convergence to the stable distribution with index $\alpha\le 1$.
    \end{enumerate}
\end{remark}

\begin{assumption}[Convergence for a sequence of random measures]\label{a2}
    Let $\{\rho_N\}_N$ be a sequence of random measures, i.e. all $a_j^{(N)}$'s are now random points. There exists $0<\eta<1$ such that the extended $\eta$-Wasserstein distance between $\rho_N$ and $\rho$ satisfies
    \begin{equation}
        \lim_{N\to+\infty}\E \dv_{W_\eta}(\rho_N,\rho)=0,\qquad 
        \sup_N\E\exp\dv_{W_\eta}(\rho_N,\rho)<+\infty.
    \end{equation}
\end{assumption}

\begin{remark}
    By the Markov inequality, the second requirement implies an exponential bound of the tail of $\dv_{W_\eta}(\rho_N,\rho)$ as
    \begin{equation}
        \mathbb P\left[\dv_{W_\eta}(\rho_N,\rho)<y\right]\le e^{-y} \sup_N\E\exp\dv_{W_\eta}(\rho_N,\rho)<+\infty
    \end{equation}
    for $y>0$.
\end{remark}

\subsubsection{Symmetric functions}
For $\lambda$ be a partition with $\lambda_1\le\ldots,\le \lambda_N$ in $\Z^N$, we denote $P_{\vec\lambda}(x_1,\ldots,x_N;q,t)$ the Macdonald polynomial (for its full definition, we refer the readers to~\cite{Macdonald_book}). The Heckman-Opdam hypergeometric function $\mathcal F_{\vec r}$ is defined by a reduction of Macdonald polynomial. For $a_1<\ldots<a_N$ in $\R^N$, in the case $q=e^{-\varepsilon}, t=q^\theta (\theta>0), \vec \lambda=\lfloor\varepsilon^{-1}\vec r\rfloor$ and $\vec x=e^{\varepsilon \vec y}$, it is
\begin{equation}\label{HO_macpoly}
    \mathcal F_{\vec r}(y_1,\ldots,y_N;\theta):=\lim_{\varepsilon\to 0^+}\frac{P_{\vec \lambda}(x_1,x_2,\ldots,x_N;q,t)}{P_{\vec \lambda}(1,t,\ldots,t^{N-1};q,t)}
\end{equation}
The multivariate Bessel function $\mathcal B_{\vec a}$ is defined by a limit of the Heckman-Opdam hypergeometric function,
\begin{equation}
    \mathcal B_{\vec a}(z_1,\ldots,z_N;\theta):=\lim_{\varepsilon\to 0^+}\mathcal F_{\varepsilon \vec a}(\varepsilon^{-1}z_1,\ldots,\varepsilon^{-1}z_N;\theta)
\end{equation}
We are particularly interested in the rank-one reduction of those symmetric functions, that is, the special case where $y_2=y_3=\ldots=y_N=0$. We denote
\begin{equation}
    P_{\vec \lambda}(x_1;N,q,t):=\frac{P_{\vec \lambda}(x_1,t,\ldots,t^{N-1};q,t)}{P_{\vec \lambda}(1,t,\ldots,t^{N-1};q,t)}
\end{equation}
as the rank-one Macdonald polynomial. Then the rank-one Heckman-Opdam hypergeometric function and the multivariate Bessel function are denoted respectively by
\begin{align}
    &\mathcal F_{\vec r}(y_1;N,\theta):=\mathcal F_{\vec r}(\,\underbrace{y_1,0,\ldots,0}_{N\text{ entries}}\,;\theta),
    \quad \mathcal B_{\vec a}(y_1;N,\theta):=\mathcal B_{\vec a}(\,\underbrace{y_1,0,\ldots,0}_{N\text{ entries}}\,;\theta).
\end{align}
We also refer the readers to Propositions~\ref{Bessel} and~\ref{Heckman-Opdam} for integal representations for the two rank-one symmetric functions.

\subsubsection{Main results}

The main results of this paper are the asymptotics of those rank-one symmetric functions in different regimes. For the multivariate Bessel function, we have the following result.

\begin{theorem}[The Fourier transform]\label{thm_main}
    Let $\{\rho_N\}$ be a sequence of discrete probability measures given in~\eqref{discrete}, and $\rho$ be its limiting measure under the convergence specified by Assumption~\ref{a1}. We also assume that $\rho\in\mathcal V$, and that for some $u\in\C\setminus\{0\}$ one has
    \begin{equation}\label{1.21}
        \int_\R e^{\re u\, |x|}\rho(\dv x)<+\infty.
    \end{equation} 
    Then the following statements hold.
    \begin{enumerate}
    \item \textbf{Classical regime:} for $\beta_N=o(1/N)$, we have
    \begin{equation}\label{1.20}
        \lim_{N\to+\infty}\mathcal B_{\vec a^{(N)}}\left(u;N,\frac{\beta_N}{2}\right)=\int_\R e^{u x}\rho(\dv x).
    \end{equation}
    \item \textbf{High-temperature regime:} for $\beta_N=\frac{2c}{N}+o(1/N)$, in the case $c>1$, we have
    \begin{equation}\label{1.24}
        \lim_{N\to+\infty}B_{\vec a^{(N)}}\left(u;N,\frac{\beta_N}{2}\right)=\int_\R e^{u x}\rho^{(c)}(\dv x).
    \end{equation}
    where $\rho^{(c)}$ and $\rho$ satisfies a Markov-Krein correspondence.
    \end{enumerate}
\end{theorem}

\begin{theorem}[The Fourier transform; the random case]\label{thm_3}
    Let $\{\rho_N\}$ be a sequence of \textbf{random} discrete probability measures given in~\eqref{discrete}, and $\rho$ be its \textbf{deterministic} limiting measure under the convergence specified by Assumption~\ref{a2}. Then in the high-temperature limit $\beta_N=\frac{2c}{N}+o(1/N)$ with $c>1$, 
    \begin{equation}\label{1.24a}
        \lim_{N\to+\infty}\E B_{\vec a^{(N)}}\left(it;N,\frac{\beta_N}{2}\right)=\int_\R e^{it x}\rho^{(c)}(\dv x),\quad t\in\R\setminus\{0\},
    \end{equation}
    where $\rho^{(c)}$ and $\rho$ satisfies a Markov-Krein correspondence, and the expectation takes over the randomness of $\vec a^{(N)}$.
\end{theorem}

\begin{remark}
    The asymptotics of the multivariate Bessel function in both the classical and high-temperature regimes can be partly recovered by~\cite{BGCG22}. In Section~\ref{s4} we will discuss the difference between their results and ours.
\end{remark}

The following theorem is the counterpart of Theorem~\ref{thm_main} for the Heckman-Opdam hypergeometric function.

\begin{theorem}[The Mellin transform]\label{thm_main_m}
    Let $\{\rho_N\}$ be a sequence of discrete probability measures on $\R_+$ given in~\eqref{discrete}, and $\rho$ be its limiting measure on $\R_+$ under the convergence specified by Assumption~\ref{a1}. We also assume that $\rho\in\mathcal V$, and that for some $u\in\C\setminus\{0\}$
    \begin{equation}
        \int_0^{+\infty} x^{\re u}\,\rho(\dv x)<+\infty.
    \end{equation}
    Then the following statements hold.
    \begin{enumerate}
    \item \textbf{Classical regime:} for $\beta_N=o(1/N)$, we have
    \begin{equation}\label{1.22}
        \lim_{N\to+\infty}\mathcal F_{\log\vec a^{(N)}}\left(u;N,\frac{\beta_N}{2}\right)=\int_0^{+\infty} x^{u}\rho(\dv x).
    \end{equation}
    \item \textbf{High-temperature regime:} for $\beta_N=\frac{2c}{N}+o(1/N)$ and $\re u\ge 1$, in the case $c>1$ we have
    \begin{equation}\label{1.26}
        \lim_{N\to+\infty}\mathcal F_{\log\vec a^{(N)}}\left(u;N,\frac{\beta_N}{2}\right)=\int_0^{+\infty} x^{u}\rho^{(c)}(\dv x),\quad \re u\ge 1.
    \end{equation}
    where $\rho^{(c)}$ and $\rho$ satisfies a Markov-Krein correspondence~\eqref{MK}.
    \end{enumerate}
\end{theorem}

\begin{remark}[Technical conditions and generalisations]\
    \begin{enumerate}
        \item For Theorem~\ref{thm_main}, Theorem~\ref{thm_3} and Theorem~\ref{thm_main_m} in the high-temperature regime, the analytic assumptions that $c>1$ is only technical. They assures the contour integral~\eqref{fourier1} (in the case $\re u=0$) to be absolutely integrable, which is a key technical point to the proof. We believe the same result holds without those assumptions, as it is implied partially by results in~\cite{BGCG22}.
        \vspace{.5em}
        \item For Theorem~\ref{thm_main_m} in the high-temperature regime, we furthermore assumed $\re u\ge 1$. This is because by Corollary~\ref{cor_tail} the existence of the Mellin transform of $\rho^{(c)}$ can only be guaranteed in this case. It remains open to extend it to all $u\in \C\setminus\{0\}$.
        \vspace{.5em}
        \item For Theorem~\ref{thm_3} we take $u=it\in i\R$. We believe this setting can be generalised to $u\in\C\setminus\{0\}$ with a similar condition~\eqref{1.21} given in Theorem~\ref{thm_main}. However, this will further requires that the supports of the measures $\rho_N$ are uniformly bounded at least from one side of the real axis. We decide to keep this case open for the sake of convenience.
        \vspace{.5em}
        \item We believe that Theorem~\ref{thm_3} can be generalised to the classical regime. To do so it requires a more detailed analysis on the error term in Section~\ref{s3.3}, and relates the difference in Stieltjes transforms to the difference in Fourier transform.
        \vspace{.5em}
        \item We believe that Theorem~\ref{thm_3} can also be generalised to the case of Mellin transform. However, our integral representation~\eqref{mellin2} requires the support of the measure to be bounded either away from infinity or away from zero (when one takes the inverse of the random variable), which in turn requires for case-by-case more detailed analysis of the difference between the Mellin transforms.
    \end{enumerate}
    
\end{remark}

\subsubsection{Strategy of the proofs}

We will mainly demonstrate how one proves~\eqref{1.20} and~\eqref{1.24}. The proofs of~\eqref{1.22} and~\eqref{1.26} are in a similar manner.

\vspace{1em}
\smallskip\noindent\emph{Step I.} We will show that for probability measures $\rho^{(c)}$ and $\rho$ satisfying the Markov-Krein correspondence, there is a contour integral formula for the characteristic function of $\rho^{(c)}$, in terms of $\rho$ (Theorem~\ref{Fourier}). This is done by investigating the theory of Dirichlet process.

\vspace{1em}
\smallskip\noindent\emph{Step II.} We show that on the other hand, the multivariate Bessel function has an integral expression (Proposition~\ref{Bessel}). This agrees with the expression above (Corollary~\ref{cor})in the following sense: when $\rho_N$ is a discrete measure, for any $c=\beta N/2$, the integral formula between the measures $\rho_N^{(c)}$ and $\rho_N$ is exactly the same as the one for the multivariate Bessel function $\mathcal B_{\vec a}(s;N,\beta/2)$, where $\vec a=(a_1,\ldots,a_N)$ gives the point mass of the discrete measure $\rho_N$. This is also a finite $N$ and discrete version of~\eqref{1.24}. This theorem is shown by generalising the existing integral formula for the multivariate Bessel function in~\cite{Cuenca21} using a contour deformation.

Now we have for $c_N=N\beta_N/2$,
\begin{equation}
    \mathcal B_{\vec a}\left(u;N,\frac{\beta_N}{2}\right)=\int_\R e^{ux}\rho_N^{(c_N)}(\dv x)
\end{equation}
holds for all $N$, the remaining task is to show that in the classical or high-temperature limit, i.e. $c_N\to 0$ or $c_N\to c>0$, the right hand side converges to the Fourier transform of $\rho^{(c)}$, that is, a weak convergence $\rho_N^{(c_N)}\to\rho^{(c)}$.

\vspace{1em}
\smallskip\noindent\emph{Step III.} To show~\eqref{1.20}, since $\rho_N^{(c_N)}$ and $\rho_N$ satisfy a Markov-Krein correspondence, we expand this correspondence for large $N$. This gives us the dominating contribution for the measure $\rho_N^{(c_N)}$ (which is of course $\rho^{(c)}$) and the error terms. A Stieltjes inversion theorem has been used to transfer convergence of the Stieltjes transforms back to the weak convergence of measures. See Section~\ref{s3.3}.

\vspace{1em}
\smallskip\noindent\emph{Step IV.} To show~\eqref{1.24}, we will directly show the convergence between the Fourier transforms of $\rho_N^{(c_N)}$ and $\rho^{(c)}$. This is done by studying the asymptotics of the characteristic function of $\rho^{(c_N)}_N$ and show that the error between it and the characteristic function of $\rho^{(c_N)}_N$ is small; see Section~\ref{s3.4}.

\vspace{1em}
\smallskip\noindent\emph{Step V.} To show the respective results for the multiplicative case, we follow Step I and II to show that
\begin{equation}
    \mathcal F_{\log\vec a}\left(u;N,\frac{\beta_N}{2}\right)=\int_0^{+\infty }x^u\rho_N^{(c_N)}(\dv x).
\end{equation}
According to the weak convergence $\rho_N^{(c_N)}\to\rho^{(c)}$, the convergence between their Mellin transforms is naturally obtained.

\vspace{1em}
\smallskip\noindent\emph{Step VI.} Finally, to show Theorem~\ref{thm_3}, one needs to investigate the difference between both sides, which has been obtained in Step IV. Applying Assumption~\ref{a2} to the difference gives the result.

\subsubsection{Organisation of the paper}

In Section~\ref{s2} we will make a brief introduction of Dirichlet process, the relationship with the Markov-Krein correspondence, and derive several new results including:
\begin{itemize}
\item a convergence result of $\rho^{(c)}$ (Corollary~\ref{weak_conv0});
\item a comparison of the tail behaviour between $\rho$ and $\rho^{(c)}$ (Theorem~\ref{thm_tail} and Corollary~\ref{cor_tail});
\item the Fourier and Mellin transform of $\rho^{(c)}$ in terms of contour integrals (Theorem~\ref{Fourier} and~\ref{Mellin}).
\end{itemize}

In Section~\ref{s3}, we will give two contour integral expressions for multivariate Bessel function and the Heckman-Opdam hypergeometric function (Proposition~\ref{Bessel} and Proposition~\ref{Heckman-Opdam}). Using these and together with~Theorem~\ref{Fourier} and~\ref{Mellin}, we provide proofs for Theorem~\ref{thm_main} and~\ref{thm_main_m}. We then end the paper with some Discussions in Section~\ref{s4}.

\section{Dirichlet process and the Markov-Krein correspondence}\label{s2}

In this section we review the theory of the Dirichlet process using the notation of~\cite{LR04,DelloSchiavo19}. Our aim is to familiarise the reader with the basic framework while also introducing several new results, namely Corollary~\ref{weak_conv0}, Theorem~\ref{thm_tail}, Theorem~\ref{Fourier}, and Theorem~\ref{Mellin}. Some of the statements from~\cite{LR04} will be recalled here for completeness. We note that a few proofs in~\cite{LR04} contain minor gaps (in particular, Theorem~\ref{Fourier} i.e.~\cite[Theorem 7]{LR04}). These issues are purely technical and do not affect the essence of the results, while for the sake of a self-contained exposition, we provide complete proofs and fill in the missing details in the Appendix.

To begin with, we recall that the Dirichlet measure is a compactly supported probability measure on $\R^N$ given by
\begin{equation}
    \mathrm{Dir}_{\alpha_1,\ldots,\alpha_N}(\dv\vec \sigma):=\frac{\prod_{j=1}^N\Gamma(\alpha_j)}{\Gamma(\alpha_1+\ldots+\alpha_N)}\ind_{\vec \sigma\in \Delta_N}\prod_{j=1}^N\sigma_j^{\alpha_j-1}\dv\sigma_j,
\end{equation}
with support given by
\begin{equation}
    \Delta_N:=\left\{  \vec\sigma\in[0,1]^N: \sum_{j=1}^N\sigma_j=1\right\}.
\end{equation}

\subsection{Dirichlet process and its random mean}\label{s2.1}

Let $\R$ equipped with the standard topology and the Borel $\sigma$-algebra $\mathcal B(\R)$. Denote $C_c(\R)$ the set of all compactly supported continuous functions on $\R$. We let $\mathcal P(\R)$ be the space of all probability measures defined on $\R$, equipped with the weak${}^*$ topology which is induced by the duality of $C_c(\R)$. Equipping it with $\mathcal B(P(\R))$ the Borel $\sigma$-algebra induced by this topology makes $(\mathcal P(\R),\mathcal B(\mathcal P(\R)))$ a measurable space of probability measures.

Let $(\Omega,\Sigma,\mathbb P)$ be the probability space. A random probability measure $\zeta$ is a probability-measure-valued random variable, that is, a measurable map between $(\Omega,\Sigma,\mathbb P)$ and $(\mathcal P(\R),\mathcal B(\mathcal P(\R)))$. Let $\rho$ be a probability measure on $\R$ and $c>0$. A Dirichlet process $D_{c\rho}$ with parameter $c\rho$ is the unique random probability measure on $X$, such that for any finite measurable partition $A_1,\ldots,A_N\in\mathcal B(\mathcal P(\R))$ ($A_1\cup\ldots \cup A_N=\R$ and $A_j\cap A_k=\emptyset$ for any $j,k$), the mapping from $\Omega$ to $\R^N$
$$\omega\mapsto (D_{c\rho}(\omega)(A_1),\ldots,D_{c\rho}(\omega)(A_N))$$ 
being an $\R^N$-valued random variable, following the Dirichlet distribution $\mathrm{Dir}_{(c\rho(A_1),\ldots,c\rho(A_N))}$ on $\Delta_N$. Here $\omega\in\Omega$ is an event and $D_{c\rho}(\omega)\in\mathcal P(X)$ is a probability measure on $X$. One also shortens the notation by writing $D_{c\rho}=D_{c\rho}(\omega)$ if it is clear from the context. The existence and uniqueness of this random measure can be verified using the Kolmogorov extension theorem; see e.g.~\cite{Ferguson73}.

The random mean of $D_{c\rho}$ is given by the law of the following $\R$-valued random variable:
\begin{equation}
    \omega\mapsto\int_\R x\,D_{c\rho}(\omega)(\dv x),
\end{equation}
where $\omega\in\Omega$ and $D_{c\rho}(\omega)\in\mathcal P(\R)$. We denote the law by $\rho^{(c)}$---clearly, $\rho^{(c)}$ depends both on $\rho$ and on $c$.

In general there is no guarantee of the existence of $\rho^{(c)}$, simply because the integral $\int_\R x\,D_{c\rho}(\omega)(\dv x)$ may not converge for all $D_{c\rho}(\omega)$. In~\cite[Theorem 1]{LR04} it is shown that this integral is absolutely integrable almost surely if and only if the logarithm function is integrable with respect to the measure $\rho$, and thus $\rho^{(c)}$ is well-defined.

\begin{theorem}[~\cite{LR04} Existence of the random mean]\label{equiv}
    Let $c>0$. A probability measure $\rho$ satisfies
    \begin{equation}\label{cond_nu}
        \int_\R\log(1+x^2)\rho(\dv x)<+\infty
    \end{equation}
    if and only if the Dirichlet process $D_{c\rho}$ satisfies
     \begin{equation}
        \mathbb P\left(\int_\R |x|\,D_{c\rho}(\dv x)<+\infty\right)=1.
    \end{equation}
\end{theorem}

\begin{remark}
    In~\cite{LR04} the condition~\eqref{cond_nu} for $\rho$ is instead $\int_\R\log(1+|x|)\rho(\dv x)<+\infty$. Though, it is easy to see the equivalence via the inequality
    \begin{equation}
        \log(1+|x|)-1<\log(1+x^2)<2\log(1+|x|)
    \end{equation}
    for $x\in\R$. We take~\eqref{cond_nu} here in order to avoid the non-analyticity of the absolute value function at $0$.
\end{remark}

Intimately related to the random mean of a Dirichlet process is a functional equation called the Markov-Krein correspondence, which gives us a bijective relation between $\rho^{(c)}$ and $\rho$. It is shown that under condition~\eqref{cond_nu}, the Markov-Krein correspondence can be satisfied by a unique choice $\rho^{(c)}$.

\begin{theorem}[\cite{LR04} Existence and Uniqueness of the Markov-Krein correspondence]\label{Exs_unq}
    For any probability measure $\rho$ satisfying~\eqref{cond_nu} and any $c>0$, there exists a unique probability measure $\rho^{(c)}$ satisfying the Markov-Krein correspondence 
    \begin{equation}\label{MKR}
    \int_\R\frac{\rho^{(c)}(\dv s)}{(z-s)^c}=\exp\left(-c\int_\R\log(z-s)\rho(\dv s)\right)
    \end{equation}
    for all $z\in\C\setminus\R$. Also this probability measure $\rho^{(c)}$ is given precisely by the random mean of the Dirichlet process $D_{c\rho}$.
\end{theorem}

We will briefly illustrate how one shows the existence of $\rho^{(c)}$ in Theorem~\ref{Exs_unq}; further details can be checked in~\cite{LR04}. The technique used to extend the finitely supported case to the general case is to use an extended Lauricella formula, which has been stated as~\cite[Eqn. (3.1)]{LR04}. Detailed proofs of Theorem~\ref{equiv} and Theorem~\ref{Exs_unq} will be given in Appendix~\ref{A1}.

\subsection{Convergence of measures} 

One aims at showing a relation between weak convergence in $\rho^{(c)}$ and in $\rho$. To do so, let us introduce a functional description (c.f.~\cite{DelloSchiavo19}) of Dirichlet process and its random mean, in comparison with the probabilistic description introduced in Section\ref{s2.1}.

Let $\zeta$ be a random measure (i.e., a measurable map from the probability space to $(\mathcal P(\R),\mathcal B(\mathcal P(\R)))$). For any measurable function $\varphi$ on $\R$ we denote the \textit{evaluation map} as
\begin{equation}
    \mathrm{ev}^\varphi:\mathcal P(\R)\mapsto \R,\quad \mathrm{ev}^\varphi(p):=\int_\R\varphi(x)p(\dv x)
\end{equation}
provided that the integral converges absolutely. For a bounded $\varphi$, one can check that $\mathrm{ev}^\varphi$ is continuous, and hence a measurable map from $(\mathcal P(\R),\mathcal B(\mathcal P(\R)))$ to $(\R,\mathcal B(\R))$. A simple composition of $\mathrm{ev}^\varphi$ and $\zeta$ gives a real-valued random variable, i.e. a measurable map from the probability space to $(\R,\mathcal B(\R))$. One denotes $\rho_\sharp f$ the push-forward measure of $\rho$ under a measurable map $f$. Then we have the following commutative diagram between measure spaces:
\[
\begin{tikzcd}
	(\Omega,\Sigma,\mathbb P) && (\mathcal P(\R),\mathcal B(\mathcal P(\R)),\zeta_\sharp\mathbb P) \\
 \\
	&& (\R,\mathcal B(\R),(\mathrm{ev}^\varphi\circ\zeta)_\sharp\mathbb P)
	\arrow["\zeta", from=1-1, to=1-3]
	\arrow["\mathrm{ev}^\varphi\circ\zeta" ', from=1-1, to=3-3]
	\arrow["\mathrm{ev}^\varphi", from=1-3, to=3-3]
\end{tikzcd}
\]
for any bounded measurable function $\varphi$. In the case $\varphi(x)=\ind_{x\in A}$ being an indicator function of a measurable set $A\subset \mathcal B(\R)$, one also denotes its evaluation map as $\mathrm{ev}^A$.

Now one introduces an alternative definition for the Dirichlet process. Let $\rho$ and $c>0$ be stated as in Section~\ref{s2.1}. A Dirichlet process $D_{c\rho}$ with parameter $c\rho$ is the unique $\mathcal P(\R)$-valued random variable such that for any finite measurable partition $A_1,\ldots,A_N\in\mathcal B(\R)$, the vector 
$$(\mathrm{ev}^{A_1}\circ D_{c\rho},\ldots,\mathrm{ev}^{A_N}\circ D_{c\rho})$$ 
is an $\R^N$-valued random variable, with its law being the Dirichlet distribution $\mathrm{Dir}_{(c\rho(A_1),\ldots,c\rho(A_N))}$ on $\Delta_N$. Then the law of the random mean $\rho^{(c)}$ is the law of the real-valued random variable $\mathrm{ev}^{\mathrm{id}}\circ D_{c\rho}$ where $\mathrm{id}$ denotes the identity function. Precisely, it is
\begin{equation}
    \rho^{(c)}:=(\mathrm{ev}^{\mathrm{id}}\circ D_{c\rho})_\sharp\mathbb P.
\end{equation}
Theorem~\ref{equiv} tells us that $\rho^{(c)}$ is well-defined, namely, the integral $\mathrm{ev}^{\mathrm{id}}(p)$ converges absolutely for any $p\in D_{c\rho}(\Omega)\subset\mathcal P(\R)$, if and only if~\eqref{cond_nu} holds for $\rho$. Thus we have the following commutative diagram.
\[
\begin{tikzcd}
	(\Omega,\Sigma,\mathbb P) && (D_{c\rho}(\Omega),\mathcal B(D_{c\rho}(\Omega)),(D_{c\rho})_\sharp\mathbb P) \\
 \\
	&& (\R,\mathcal B(\R),(\mathrm{ev}^\mathrm{id}\circ D_{c\rho})_\sharp\mathbb P)
	\arrow["D_{c\rho}", from=1-1, to=1-3]
	\arrow["\mathrm{ev}^\mathrm{id}\circ D_{c\rho}" ', from=1-1, to=3-3]
	\arrow["\mathrm{ev}^\mathrm{id}", from=1-3, to=3-3]
\end{tikzcd}
\]

\begin{theorem}[\cite{DelloSchiavo19}]
    For any $c>0$, the map $\mathcal D_c:\mathcal P(\R)\mapsto\mathcal P(\mathcal P(\R)), \mathcal D_c(\rho)=D_{c\rho}$ is continuous with respect to the narrow topologies on $\mathcal P(\R)$ and $\mathcal P(\mathcal P(\R))$. The narrow convergence of the latter is induced by duality of bounded continuous functions on $\mathcal P(\R)$.
\end{theorem}
Using this one is able to prove the continuity of the map $\rho\mapsto\rho^{(c)}$ for any given $c>0$.

\begin{corollary}\label{weak_conv0}
    Take any $c>0$. Let $\{\rho_n\}$ be a sequence of probability measure satisfying~\eqref{cond_nu}, and let $\rho^{(c)}_n$ be the random mean of $D_{c\rho_n}$. Let the weak convergence $\rho_n\to\rho$ holds as $n$ goes to infinity, and let the limit $\rho$ also satisfy~\eqref{cond_nu}, then there is a weak convergence $\rho^{(c)}_n\to\rho^{(c)}$ with the measure $\rho^{(c)}$ being the random mean of $D_{c\rho}$.
\end{corollary}

\begin{proof}
    The fact that both $\rho_n$ and $\rho$ satisfy~\eqref{cond_nu} guarantees the existence of the Dirichlet processes $D_{c\rho_n}$ and $D_{c\rho}$. We take an arbitrary bounded continuous function $\varphi:\R\to\R$. Then by the property of the push-forward measure, for any $\rho_n$  satisfying~\eqref{cond_nu}, we have
    \begin{align*}
        \int_\R\varphi(x)\,\rho^{(c)}_n(\dv x)&=\int_\R\varphi(x)\,(\mathrm{ev}^{\mathrm{id}}\circ D_{c\rho_n})_\sharp\mathbb P(\dv x)\\
        &=\int_{\mathcal P(\R)}\varphi(\mathrm{ev}^{\mathrm{id}}(p))\,(D_{c\rho_n})_\sharp\mathbb P(\dv p)\\
        &=\int_{\Omega}\varphi((\mathrm{ev}^{\mathrm{id}}\circ D_{c\rho_n})(\omega))\,\mathbb P(\dv \omega).
    \end{align*}
    By Theorem~\ref{equiv} we know that when $\rho_n$ satisfies~\eqref{cond_nu}, $(\mathrm{ev}^{\mathrm{id}}\circ D_{c\rho_n})(\omega)<+\infty$ almost surely, and thus $\varphi(\mathrm{ev}^{\mathrm{id}}(p))$ is almost surely a bounded function in $p$ (the continuity of this function can be verified easily). Then by~\cite[Thm 1.1]{DelloSchiavo19} one has
    \begin{equation}
        \lim_{N\to+\infty}\int_{\mathcal P(\R)}\varphi(\mathrm{ev}^{\mathrm{id}}(p))\,(D_{c\rho_n})_\sharp\mathbb P(\dv p)=\int_{\mathcal P(\R)}\varphi(\mathrm{ev}^{\mathrm{id}}(p))\,(D_{c\rho})_\sharp\mathbb P(\dv p),
    \end{equation}
    which is the weak convergence $\rho^{(c)}_n\mapsto\rho^{(c)}$.    
\end{proof}

\subsection{Further properties of the Markov-Krein correspondence}
We denote the sets
\begin{align}
    &\mathcal V\,\, :=\left\{\rho:\int_\R\log(1+x^2)\rho(\dv x)<+\infty\right\},\label{2.27}\\
    &\mathcal V_c:=\left\{\rho^{(c)}:\rho^{(c)}\text{ is the distribution for }\int x\,D_{c\rho}(\dv x);\, \rho\in\mathcal V\right\}.
\end{align}
The following theorem shows a bijective relation between $\mathcal V$ and $\mathcal V_c$. The proof can be found in~\cite[Theorem 2]{LR04}.
\begin{theorem}[\cite{LR04}]
    For a given $c>0$, the Markov-Krein correspondence defines a bijection between $\mathcal V$ and $\mathcal V_c$.
\end{theorem}

We will show a new result stating that the tail behaviour of $\rho^{(c)}$ is governed by the one for $\rho$.

\begin{theorem}\label{thm_tail}
    Let $\rho$ be a probability measure and $\rho^{(c)}$ be the random mean of the Dirichlet process $D_{c\rho}$, with some $c>0$. Let $\varphi$ be a non-negative, even and convex function, i.e.
    \begin{equation}\label{K_ineq}
        \varphi\left(\sum_{j=1}^n\sigma_jx_j\right)\le \sum_{j=1}^n\sigma_j\varphi\left(x_j\right),
    \end{equation}
    and for all $\sigma_1,\ldots,\sigma_n\in[0,1], \sigma_1+\ldots+\sigma_n=1$ and $x_1,\ldots,x_n\in\R$. We furthermore assume that there is an even and positive function $h(x)$ which is increasing when $x$ is large enough, such that
    \begin{equation}
        \int h(x)\rho(\dv x)<+\infty,\quad h(x)\ge \varphi(x).
    \end{equation}
    Then one has
    \begin{equation}
        \int_\R\varphi(x)\rho^{(c)}(\dv x)\le \frac{1}{c}\int_\R\varphi(x)\rho(\dv x)<+\infty.
    \end{equation}
\end{theorem}

\begin{proof}[Proof of Theorem~\ref{thm_tail}]
    We firstly show this result for the case when $\rho$ is finitely supported, i.e. it is given by $\rho=\sum_{j=1}^Na_j\delta_{x_j}$ for $x_1,\ldots,x_n\in\R$ and $a_1,\ldots,a_n\in[0,1]$ with $a_1+\ldots+a_n=1$. Substituting~\eqref{K_ineq} into~\eqref{Dir2}, we see that 
    \begin{equation}\label{2.38}
        \int\varphi(x)\rho^{(c)}(\dv x))=\int\varphi\left(\sum_{j=1}^n\sigma_jx_j\right)\mathrm{Dir}_{ca_1,\ldots,ca_n}(\dv\sigma)\le\sum_{j=1}^n K\varphi(x_j)\int\sigma_j\mathrm{Dir}_{ca_1,\ldots,ca_n}(\dv\sigma).
    \end{equation}
    One computes that the mean of a Dirichlet distribution $\mathrm{Dir}_{ca_1,\ldots,ca_n}$ is given by $a_j/c$. Thus the previous inequality can be rewritten as
    \begin{equation}\label{discrete_case}
        \int\varphi(x)\rho^{(c)}(\dv x)\le \frac{1}{c}\sum_{j=1}^n a_j\varphi(x_j)=\frac{1}{c}\int\varphi(x)\rho(\dv x)
    \end{equation}
    which holds for any discrete measure $\rho$.

    Now for a general probability measure $\rho$, we approximate it using a sequence of discrete measures $\rho_n$ with the index set $J_n$
    \begin{equation}
        \rho_n(\dv x)=\sum_{j\in J_n}a_j\delta_{x_j}(\dv x),
    \end{equation}
    such that $x_1=-n, x_{|J_n|}=n$ and that $x_j$ is given uniquely by
    \begin{equation}
        \rho((-n,x_j])=\frac{j-1}{|J_n|-1}\rho((-n,n]).
    \end{equation}
    We take $a_j=\rho((x_{j-1},x_j])$ for $j=1,\ldots,|J_n|-1$. Then $\rho_n$ is a discrete approximation of $\rho \ind_{[-n,n]}$. It can be checked that
    \begin{equation}
    \begin{split}
        \left|\int_{(-n,n]}\varphi(x)(\rho_n(\dv x)-\rho(\dv x))\right|&\le\frac{\sup_{|x|\le n}|h(x)|}{|J_n|-1}\rho((-n,n])\\
        &\le \frac{h(n)}{|J_n|-1}.
    \end{split}
    \end{equation}
    Using a Markov-inequality one has
    \begin{equation}
    \begin{split}
        \left|\int_{\R\setminus(-n,n]}\varphi(x)(\rho_n(\dv x)-\rho(\dv x))\right|&=\left|\int_{\R\setminus(-n,n]}\varphi(x)\rho(\dv x)\right|
        \\
        &\le\frac{\int_\R h(x)\rho(\dv x)}{h(n)}.
    \end{split}
    \end{equation}
    Thus taking $J_n=n\lceil h(n)\rceil$ and combining the previous two inequalities, we obtain
    \begin{equation}\label{sup_bound}
        \sup_n\left|\int_{\R}\varphi(x)(\rho_n(\dv x)-\rho(\dv x))\right|\le\sup_n\left|\frac{h(n)}{n\lceil h(n)\rceil-1}+\frac{\int_\R h(x)\rho(\dv x)}{h(n)}\right|\to 0
    \end{equation}
    as $N\to+\infty$.

    By Corollary~\ref{weak_conv0} we see that $\rho^{(c)}$ is the weak limit of $\rho^{(c)}_n$ (which is the distribution of the random mean of the Dirichlet process with parameter measure $c\rho_n$). Therefore, one has
    \begin{equation}
    \begin{split}
        \int_{\R}\varphi(x)\rho^{(c)}(\dv x)&=\lim_{R\to\infty}\lim_{N\to+\infty}\int_{-R}^R\varphi(x)\rho^{(c)}_n(\dv x)
        \\&\le \lim_{R\to\infty}\limsup_{N\to+\infty}\int_{-R}^R\varphi(x)\rho^{(c)}_n(\dv x)\\&\le\lim_{R\to\infty}\limsup_{N\to+\infty}\int_\R\varphi(x)\rho^{(c)}_n(\dv x)=\limsup_{N\to+\infty}\int_\R\varphi(x)\rho^{(c)}_n(\dv x).
    \end{split}
    \end{equation}
    By~\eqref{discrete_case}, one has that
    \begin{equation}
        \limsup_{N\to+\infty}\int\varphi(x)\rho^{(c)}_n(\dv x)\le \frac{1}{c}\lim_{N\to+\infty}\int\varphi(x)\rho_n(\dv x)= \frac{1}{c}\int\varphi(x)\rho(\dv x),
    \end{equation}
    where the equation is due to~\eqref{sup_bound} and Lebesgue's dominated convergence theorem. Combining the above inequalities gives the final result, where the existence of the left side can be verified simply by the monotone convergence theorem.
\end{proof}

\begin{corollary}\label{cor_tail}
    We have the following bounds: for the exponential function
    \begin{equation}\label{tail_exp}
        \E_{\rho^{(c)}}[\exp(a|X|)]\le \frac{1}{c}\E_{\rho}[\exp(a|X|)],\quad a>0;
    \end{equation}
    and for the power function we have
    \begin{equation}\label{tail_polyn}
        \E_{\rho^{(c)}}[|X|^\alpha]\le\frac{1}{c}\E_{\rho}[|X|^\alpha],\quad \alpha\ge 1.
    \end{equation}
    In the above inequalities the right sides are not necessarily finite.
\end{corollary}
\begin{proof}
    We only need to show the inequalities where the right sides are finite. Since both $e^{a|x|}$ and $|x|^{\alpha}$ are non-negative even convex functions, the inequality~\ref{K_ineq} holds and thus by Theorem~\ref{thm_tail} we have the result.
\end{proof}

Another interesting property for $\rho^{(c)}$ is the relation between the supports. The original proof of the following theorem is given in~\cite{Majumdar92}. 
\begin{theorem}[~\cite{Majumdar92}]\label{conv_hull}
    For a given $c>0$ and $\rho\in\mathcal V$, the support of $\rho^{(c)}$ is a convex hull of the support of $\rho$.
\end{theorem}

\subsection{Fourier transform of \texorpdfstring{$\rho^{(c)}$}{TEXT}}\label{sec:Fourier}

Fourier transform of $\rho^{(c)}$ has been extensively studied in~\cite{LR04,Yamamoto84,DelloSchiavo19}. We will show a contour integral representation for the Fourier transform of $\rho^{(c)}$, which will be useful in the proofs of our main theorems.

We say the contour $\mathcal C$ is the Hankel loop, if it starts and ends at $-\infty$, enclosing $(-\infty,x]$ oriented counter-clockwise. A typical choice is that, for some $M>x$ it goes horizontally from $-\infty-iM$ to $M-iM$, then vertically to $M+iM$, and then horizontally to $-\infty+iM$; see Fig.~\ref{contour1}. 
\begin{figure}[t]
\begin{center}
\begin{tikzpicture}[
    my contour/.style={thick, 
        postaction={decorate}, 
        decoration={markings, 
            mark=at position 0.15 with {\arrow{Stealth}},
            mark=at position 0.55 with {\arrow{Stealth}},
            mark=at position 0.88 with {\arrow{Stealth}}
        }
    },
    axis/.style={-{Stealth[scale=1.5]}, semithick},
]

\def\infOffset{3.5}
\def\BigM{1}

\draw[my contour] 
    (-\infOffset, -\BigM) node[below left] {\scriptsize $-\infty - iM$} 
    -- 
    (\BigM, -\BigM) coordinate (bottom) 
    -- 
    (\BigM, \BigM) coordinate (top) 
    -- 
    (-\infOffset, \BigM) node[above left] {\scriptsize $-\infty + iM$};

\node[below right] at (bottom) {\scriptsize $M - iM$};
\node[above right] at (top) {\scriptsize $M + iM$};

\draw[axis] (-\infOffset-0.5, 0) -- (\BigM+1.5, 0) node[below] {$\mathrm{Re}(z)$};
\draw[axis] (0, -\BigM-1) -- (0, \BigM+1) node[left] {$\mathrm{Im}(z)$};
\node at (0,0) [below left] {\scriptsize $0$};

\end{tikzpicture}
\caption{A standard choice for the Hankel loop $\mathcal C$.}
\label{contour1}
\end{center}
\end{figure}
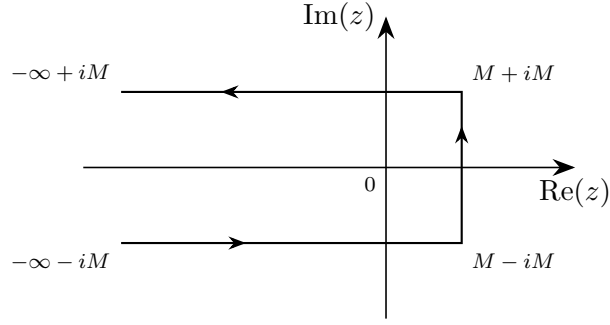

Also, we denote a sign function for complex numbers as
\begin{equation}\label{sgn}
    \sgn(u):=
    \begin{cases}
        1,&\re{u}>0;\\
        -1,&\re{u}<0;\\
        -i,&\re u=0\text{ and }\im u>0;\\
        i,&\re u=0\text{ and }\im u<0.
    \end{cases}
\end{equation}

\begin{theorem}\label{Fourier}
    Let $u\in\C\setminus\{0\}$. Assuming that 
    \begin{equation}
        \int_\R e^{x\re u }\, \rho(\dv x)<+\infty,
    \end{equation}
    then
    \begin{equation}\label{fourier1}
        \int_\R \rho^{(c)}(\dv x)e^{x u}=\lim_{M\to+\infty}\frac{\Gamma(c)}{u^{c-1}}\frac{1}{2\pi i}\int_{\sgn(u)\mathcal C}\dv z\,\exp\left(u z-c\int_\R\log(z-s)\rho^{[M]}(\dv s)\right)
    \end{equation}
    and the left side converges. The contour $\mathcal C$ is the Hankel loop enclosing the branch cut $(-\infty,M]$. If $\rho$ is compactly supported in $[-M,M]$, one has instead
    \begin{equation}\label{fourier2}
        \int_\R \rho^{(c)}(\dv x)e^{x u}=\frac{\Gamma(c)}{u^{c-1}}\frac{1}{2\pi i}\int_{\sgn(u)\mathcal C}\dv z\,\exp\left(u z-c\int_\R\log(z-s)\rho(\dv s)\right).
    \end{equation}
    
   In the special case $u=it$ and $t\ne 0$, the contour $\sgn(u)\mathcal C$ in~\eqref{fourier1} and~\eqref{fourier2} can be deformed to $\R-i\delta \sgn(t)$ for any $\delta>0$. If furthermore $c>1$, Equation~\eqref{fourier2} also holds for non-compactly supported $\rho$.
\end{theorem}

Theorem~\ref{Fourier} might be formally derived from a slightly different form in~\cite[Theorem 7]{LR04} using a change of variable, but in the arguments of~\cite{{LR04}} the convergence of integrals have not been taken care of. We will provide an detailed proof in Appendix~\ref{A:Fourier}, taking care of all issues of convergence.

\subsection{Mellin transform of \texorpdfstring{$\rho^{(c)}$}{rho(c)}}

Similar to the Fourier transform, the Mellin transform of $\rho^{(c)}$ can also be rewritten into a contour integral formula, which has not been discussed in the literature.

\begin{lemma}\label{lem_Beta}
    For all $c>0$, $u\in\C\setminus \{0\}$ with $\re u>0$, $x>0$ and $t\in\R\setminus\{0\}$, we have
    \begin{equation}\label{beta}
    x^{u}=\frac{\Gamma(c)\Gamma(u+1)}{\Gamma(u+c)}\frac{1}{2\pi i}\int_{\mathcal C}\dv z\,\frac{e^{uz}}{(1-xe^{-z})^{c}}
    \end{equation}
    where the contour $\mathcal C$ is the Hankel loop enclosing $(-\infty, \log x]$.
\end{lemma}

\begin{proof}
    Let us recall the contour integral representation for Beta function
    \begin{equation}
        \frac{1}{bB(a,b)}=\frac{1}{2\pi i}\int_{\gamma+i\R}\dv t\, t^{-a}(1-t)^{-1-b},\quad 0<\gamma<1, \re (a+b)>0.
    \end{equation}
    Taking the change of variables $1-t\mapsto xe^{-z}$, $a\mapsto c$, $b\mapsto u$ gives
    \begin{equation}
        \frac{1}{uB(c,u)}=\frac{1}{2\pi i}\int_{\mathcal C'}\dv z\, (1-xe^{-z})^{-c}x^{-u}e^{uz},\quad \re (u+c)>0.
    \end{equation}
    One can check that the new contour $\mathcal C'$ is given by the set of points $z$ satisfying
    \begin{equation}
        \re (1-xe^{-z})=\gamma\quad\Longleftrightarrow \quad e^{-\re z}\cos \im z=\frac{1-\gamma}{x},
    \end{equation}
    and thus it is a loop enclosing the branch cut at $(-\infty,\log x]$ counter-clockwise with asymptotes $\im z=\pm \frac{\pi}{2}$. We also observe that the integrand is analytic except for the branch cut. At $-\infty$, we can check that for large $R$ and any $M>0$,
    \begin{equation}
    \begin{split}
        \left|\int_{-R-Mi}^{-R+Mi}\dv z \frac{e^{uz}}{(1-x e^{-z})^c}\right|\le \int_{-M}^M \dv y\frac{e^{-R\re u-y\im u}}{|1-xe^{R}|^c}\le \frac{e^{-R\re u}e^{M|\im u|}}{|1-xe^{R}|^c}\to 0
    \end{split}
    \end{equation}
    as $R\to\infty$. Therefore the contour can be deformed to the Hankel loop, for example the one given in Lemma~\ref{lem_Gamma} (provided $M>\log x$).
\end{proof}

\begin{theorem}\label{Mellin}
    Let $y\in\C$ and $\re u>0$. Assuming that 
    \begin{equation}
        \int_{0}^{+\infty} x^{\re u}\, \rho(\dv x)<+\infty,
    \end{equation}
    then
    \begin{equation}\label{mellin1}
        \int_\R \rho^{(c)}(\dv x)\,x^{u}=\lim_{M\to+\infty}\frac{\Gamma(c)\Gamma(u+1)}{\Gamma(u+c)}\frac{1}{2\pi i}\int_{\mathcal C}\dv z\,\exp\left(u z-c\int_\R\log(1-se^{-z})\rho^{[M]}(\dv s)\right)
    \end{equation}
    and the left side converges. If $\rho$ is compactly supported, one has instead
    \begin{equation}\label{mellin2}
        \int_\R \rho^{(c)}(\dv x)\,x^{y}=\frac{\Gamma(c)\Gamma(u+1)}{\Gamma(u+c)}\frac{1}{2\pi i}\int_{\mathcal C}\dv z\,\exp\left(u z-c\int_\R\log(1-se^{-z})\rho(\dv s)\right).
    \end{equation}
    Here $\mathcal C$ is the Hankel loop enclosing $(-\infty,\log M]$.
\end{theorem}

\begin{proof}
    The proof is similar to the one for Theorem~\ref{Fourier}. By Theorem~\ref{thm_tail}, the left side of~\eqref{mellin1} is finite. Let $\rho$ be compactly supported, with support contained in $[0,M]$ for some $M>0$. By Theorem~\ref{conv_hull}, the support of $\rho^{(c)}$ is also contained in $[0,M]$. By substituting~\eqref{beta} into the Mellin transform of $\rho^{(c)}$, one has
    \begin{equation}
        \int_{0}^M \rho^{(c)}(\dv x)\,x^{y}=\int_{0}^M\rho^{(c)}(\dv x)\frac{\Gamma(c)\Gamma(y+1)}{\Gamma(y+c)}\frac{1}{2\pi i}\int_{\mathcal C}\dv z\,\frac{e^{uz}}{(1-xe^{-z})^{c}}
    \end{equation}
    for some Hankel loop $\mathcal C$ enclosing $(-\infty,\log M]$. By the exponential decay of $e^{uz}$ along $\mathcal C$ (because $\re u>0$) and that $|1-xe^{-z}|$ is bounded from below, the Fubini theorem can be applied to interchange the two integrals. One also recalls that by~\eqref{MKR'}
    \begin{equation}
        \int_{0}^M\frac{\rho^{(c)}(\dv x)}{(1-xe^{-z})^c}=\exp\left(-c\int_\R\log(1-se^{-z})\rho(\dv s)\right),
    \end{equation}
    and thus we obtain~\eqref{mellin2} for compactly supported $\rho$. Equation~\eqref{mellin1} can be again obtained by Lemma~\ref{weak_conv} and the Lebesgue dominated convergence theorem.

    For general $\rho$, we have the weak convergence $(\rho^{[M]})^{(c)}\to\rho^{(c)}$ by Lemma~\ref{weak_conv}, which  finalises the proof.
\end{proof}

\section{Proof of the main theorems}\label{s3}

\subsection{Integral representations for symmetric functions}\label{s3.1}

To prove the main theorems, the first step is to investigate integral representations of the symmetric functions. It is found in~\cite{Cuenca18} that in the case $|x|\ge 1$, a particular ratio of Macdonald polynomials can be rewritten as a contour integral
\begin{equation}\label{HO_integral}
\begin{split}
    \frac{P_\lambda(x,t,t^2,\ldots,t^{N-1};q,t)}{P_\lambda(1,t,t^2,\ldots,t^{N-1};q,t)}=&\frac{\log q}{q-1}\frac{(x^{-1}t^N;q)_\infty}{(x^{-1}q;q)_\infty}\frac{\Gamma_q(\theta N)}{2\pi i}\times
    \\
    &\int_{\mathcal C}\dv z\, x^z\prod_{j=1}^N\frac{\Gamma_q(z-\lambda_j+\theta j-\theta))}{\Gamma_q(z-\lambda_j-\theta j)}.
\end{split}
\end{equation}
A consequence of this, shown in~\cite[Thm 5.1 and Thm 5.2]{Cuenca21}, is an integral representation of the multivariate Bessel function. The following theorem extends the two theorems in~\cite{Cuenca21} to the case $\re u=0$.

\begin{proposition}[for $\re u\ne 0$ see Thm 5.1 and Thm 5.2 in~\cite{Cuenca21}]\label{Bessel}
    Let $\theta>0$, $N\in\N$, $a_1\le\ldots\le a_N$ in $\R$ and $u\in\C\setminus\{0\}$. Then
    \begin{equation}
        \mathcal B_{\vec a}(u;N,\theta)=\frac{\Gamma(\theta N)}{u^{N\theta-1}}\frac{1}{2\pi i}\int_{\sgn(u)\mathcal C} \dv z\,e^{uz}\prod_{j=1}^N(z-a_j)^{-\theta},
    \end{equation}
    where $\sgn(u)$ is given by~\eqref{sgn}
\end{proposition}

\begin{proof}
    The case $\re u\ne 0$ has been proved in~\cite{Cuenca21}. It remains to check the case $\re u=0$. In the proof we take the following contour $\mathcal C$: it starts at $-\infty-Mi$ going horizontally to $M-Mi$, then vertically to $M+Mi$ and then horizontally to $-\infty+Mi$, where $\{a_1,\ldots,a_M\}\subset[-M,M]$.

    We let $u=\varepsilon+it$ where $\varepsilon>0$ and $t>0$. One aims at deforming the contour $\mathcal C$ to $-i\mathcal C$. Taking $R>0$ large enough, we let $\mathcal C_1$ be the arc centered at $M+iM$ connecting $M+iM+iR$ and $M+iM-R$, and $\mathcal C_2$ be the arc centered at $-M+iM$ connecting $-M+iM+iR$ and $-M+iM-R$. Then for any small $\delta>0$, there exists a large enough $R$ such that
    \begin{equation}
    \begin{split}
        \left|\int_{\mathcal C_1}\dv z\,e^{uz}\prod_{j=1}^N(z-a_j)^{-\theta}\right|
        &\le \int_{\pi/2}^\pi R\dv \theta\frac{e^{(\varepsilon-t)M}e^{R(\varepsilon \cos\theta-t\sin\theta)}}{|R-3M|^{\theta N}}\\
        &\le \frac{\pi}{2}\frac{e^{(\varepsilon-t)M}}{|R-3M|^{\theta N}}\left(\frac{e^{\varepsilon R}}{t R}+\delta\right)\to 0
    \end{split}
    \end{equation}
    holds as $R\to\infty$, where the final bound is obtained from applying Lemma~\ref{lem_laplace}. A similar inequality holds for the contour $\mathcal C_2$ as
    \begin{equation}
        \left|\int_{\mathcal C_2}\dv z\,e^{uz}\prod_{j=1}^N(z-a_j)^{-\theta}\right|\le \frac{\pi}{2}\frac{e^{-(\varepsilon-t)M}}{|R-3M|^{\theta N}}\left(\frac{e^{\varepsilon R}}{t R}+\delta\right)\to 0.
    \end{equation}
    Thus by Cauchy's theorem,
    \begin{equation}
    \begin{split}
        \int_{\mathcal C} \dv z\,e^{(\varepsilon+it)z}\prod_{j=1}^N(z-a_j)^{-\theta}&=\left(\int_{-i\mathcal C}+\lim_{R\to\infty}\int_{\mathcal C_1}+\lim_{R\to\infty}\int_{\mathcal C_2}\right)\dv z\,e^{(\varepsilon+it)z}\prod_{j=1}^N(z-a_j)^{-\theta}\\
        &=\int_{-i\mathcal C}\dv z\,e^{(\varepsilon+it)z}\prod_{j=1}^N(z-a_j)^{-\theta}.
    \end{split}
    \end{equation}
    It remains to check that
    \begin{equation}
        \int_{-i\mathcal C}\dv z\,(e^{\varepsilon z}-1)e^{itz}\prod_{j=1}^N(z-a_j)^{-\theta}\to 0
    \end{equation}
    as $\varepsilon\to 0$. This can be verified by using the dominated convergence theorem, since $e^{itz}$ decays exponentially fast on the contour $-i\mathcal C$. Thus, by continuity of multivariate Bessel function we have
    \begin{equation}
        \mathcal B_{\vec a}(it;N,\theta)=\lim_{\varepsilon\to 0^+}\mathcal B_{\vec a}(\varepsilon+it;N,\theta)=\frac{\Gamma(\theta N)}{(it)^{N\theta-1}}\frac{1}{2\pi i}\int_{-i\mathcal C}\dv z\,e^{itz}\prod_{j=1}^N(z-a_j)^{-\theta}.
    \end{equation}

    For the case $y=\varepsilon+it$ where $\varepsilon>0$ and $t<0$, we deform the contour to $i\mathcal C$. Taking $R>0$ large enough, we let $\mathcal C_1'$ be the arc centered at $M-iM$ connecting $M-iM-iR$ and $M-iM-R$, and $\mathcal C_2'$ be the arc centered at $-M+iM$ connecting $-M+iM-iR$ and $-M+iM-R$. The two integrals on $\mathcal C_1$ and $\mathcal C_2$ is small for large $R$, and so similar analysis applies to have
    \begin{equation}
    \begin{split}
        \mathcal B_{\vec a}(it;N,\theta)=\lim_{\varepsilon\to 0^+}\frac{\Gamma(\theta N)}{(it)^{N\theta-1}}\frac{1}{2\pi i}\int_{\mathcal C} \dv z\,e^{(\varepsilon+it)z}\prod_{j=1}^N(z-a_j)^{-\theta}
        &=\frac{\Gamma(\theta N)}{(it)^{N\theta-1}}\frac{1}{2\pi i}\int_{i\mathcal C}\dv z\,e^{itz}\prod_{j=1}^N(z-a_j)^{-\theta}
    \end{split}
    \end{equation}
    which finishes the proof.
\end{proof}

Let us turn to the integral expression for the Heckman-Opdam hypergeometric function. The following integral expression has been shown in~\cite{MP20}. We attach the proof here for completeness, taking care of convergence of the integrals.

\begin{proposition}[see Section 3.1 in~\cite{MP20}]\label{Heckman-Opdam}
    Let $\theta>0$, $N\in\N$, $0<a_1\le\ldots\le a_N$ in $\R_+$ and $u\in\C$ with $\re y>0$. Then
    \begin{equation}
        \mathcal F_{\log\vec a}(u;N,\theta)=\frac{\Gamma(\theta N)\Gamma(u+1)}{\Gamma(u+\theta N)}\frac{1}{2\pi i}\int_{\mathcal C} \dv z\,e^{uz}\prod_{j=1}^N(1-a_je^{-z})^{-\theta},
    \end{equation}
    where $\mathcal C$ is the Hankel loop enclosing $(-\infty,a_N]$.
\end{proposition}

\begin{proof}
    By~\eqref{HO_macpoly} and \eqref{HO_integral} one has
    \begin{equation}
    \begin{split}
        \mathcal F_{\log\vec a}(u;N,\theta)=&\lim_{\varepsilon\to 0^+}\frac{-\varepsilon}{e^{-\varepsilon}-1}\frac{(e^{-\varepsilon (u+\theta N)};e^{-\varepsilon})_\infty}{(e^{-\varepsilon (u+1)};e^{\varepsilon})_\infty}\frac{\Gamma_{e^{-\varepsilon}}(\theta N)}{2\pi i}\times\\
        &\int_{-\mathcal C}\dv z\,e^{\varepsilon uz }\prod_{j=1}^N\frac{\Gamma_{e^{-\varepsilon}}(z-(\lfloor \varepsilon^{-1}\log a_j\rfloor-\theta j+\theta))}{\Gamma_{e^{-\varepsilon}}(z-(\lfloor \varepsilon^{-1}\log a_j\rfloor-\theta j))}.
    \end{split}
    \end{equation}
    Here the contour $\mathcal C$ encloses the negative real line. We make a change of variable $\varepsilon z\mapsto z$. Using the expression $\Gamma_q(z)=(1-q)^{1-z}\frac{(q;q)_\infty}{(q^z;q)_\infty}$, we have
    \begin{equation}
        \frac{\Gamma_{e^{-\varepsilon}}(\varepsilon^{-1}z-(\lfloor \varepsilon^{-1}\log a_j\rfloor-\theta j+\theta))}{\Gamma_{e^{-\varepsilon}}(\varepsilon^{-1}z-(\lfloor \varepsilon^{-1}\log a_j\rfloor-\theta j))}=(1-e^{-\varepsilon})^{\theta}\frac{(e^{-(z-\varepsilon\lfloor \varepsilon^{-1}\log a_j\rfloor)}e^{-\varepsilon\theta(1-j)};e^{-\varepsilon})_\infty}{(e^{-(z-\varepsilon\lfloor \varepsilon^{-1}\log a_j\rfloor)}e^{-\varepsilon\theta(-j)};e^{-\varepsilon})_\infty}.
    \end{equation}
    Denoting $x_j(\varepsilon):=z-\varepsilon\lfloor \varepsilon^{-1}\log a_j\rfloor$, the above integral simplifies to
    \begin{equation}
    \begin{split}
        \mathcal F_{\log\vec a}(u;N,\theta)=&\lim_{\varepsilon\to 0^+}\frac{-\varepsilon}{e^{-\varepsilon}-1}(1-e^{-\varepsilon})^{\theta N}\frac{(e^{-\varepsilon (u+\theta N)};e^{-\varepsilon})_\infty}{(e^{-\varepsilon (y+1)};e^{\varepsilon})_\infty}\frac{\Gamma_{e^{-\varepsilon}}(\theta N)}{2\pi i}\times\\
        &\frac{1}{\varepsilon}\int_{\mathcal C}\dv z\,e^{-uz}\prod_{j=1}^N\frac{(e^{-x_j(\varepsilon)}e^{-\varepsilon\theta(1-j)};e^{-\varepsilon})_\infty}{(e^{-x_j(\varepsilon)}e^{-\varepsilon\theta(-j)};e^{-\varepsilon})_\infty}.
    \end{split}
    \end{equation}
    It is immediate that $\Gamma_{e^{-\varepsilon}}(\theta N)\to \Gamma(\theta N)$ as $\varepsilon\to 0^+$. Noting that $(1-q)^{a-b}\frac{(q^a;q)_\infty}{(q^b;q)_\infty}\to\frac{\Gamma(b)}{\Gamma(a)}$ in the $q\to 1^-$ limit, we have
    \begin{equation}\label{3.20}
        \lim_{\varepsilon\to 0^+}(1-e^{-\varepsilon})^{\theta N-1}\frac{(e^{-\varepsilon (u+\theta N)};e^{-\varepsilon})_\infty}{(e^{-\varepsilon (u+1)};e^{\varepsilon})_\infty}=\frac{\Gamma(u+1)}{\Gamma(u+\theta N)}.
    \end{equation}
    Also, noting the limits $\frac{(wq^a;q)_\infty}{(wq^b;q)_\infty}\to (1-w)^{b-a}$ as $q\to 1^-$ and $x_j(\varepsilon)\to z-\log a_j$ as $\varepsilon\to0^+$, we have
    \begin{equation}\label{3.21}
        \lim_{\varepsilon\to 0^+}\frac{(e^{-x_j(\varepsilon)}e^{-\varepsilon\theta(1-j)};e^{-\varepsilon})_\infty}{(e^{-x_j(\varepsilon)}e^{-\varepsilon\theta(-j)};e^{-\varepsilon})_\infty}=(1-a_je^{-z})^{-\theta}.
    \end{equation}
    Consequently, on the left side of the previous equation, the ratio of q-Pochhammer symbols are absolutely bounded (for a small enough $\varepsilon$). Thus the integral is absolutely bounded, because of its exponential decay along the contour. Assembling~\eqref{3.20},~\eqref{3.21} gives the result.
\end{proof}

Combining Propositions~\ref{Bessel}, \ref{Heckman-Opdam}, \ref{Fourier}, and~\ref{Mellin}, we obtain the following corollary.

\begin{corollary}\label{cor}
    Let $\rho_N=\frac{1}{N}\sum_{j=1}^N\delta_{a_j}$, and let $\rho_N^{(c_N)}$ be the mean of the Dirichlet process with the discrete parameter measure 
    $c_N\rho_N$ for some $c_N=\theta N,\theta>0$. Then
    \begin{equation}
        \mathcal B_{\vec a}(u;N,\theta)=\int_\R e^{ux}\rho_N^{(c_N)}(\dv x).
    \end{equation}
    If furthermore, $a_1,\ldots,a_N\in\R_+$, then
    \begin{equation}
        \mathcal F_{\log\vec a}(u;N,\theta)=\int_{0}^{+\infty }x^{u}\rho_N^{(c_N)}(\dv x).
    \end{equation}
    In other words, the multivariate Bessel function and the Heckman-Opdam hypergeometric function exactly equal to the Fourier and Mellin transforms of the measure $\rho_N^{(c_N)}$.
\end{corollary}

With the previous corollary, the main theorems are to show that in the classical regime, the measure $\rho_N^{(c_N)}$ converges weakly to $\rho$; in the high-temperature regime, $\rho_N^{(c_N)}$ converges weakly to $\rho^{(c)}$ instead.

\subsection{Approximating measures}\label{s3.2}
    Before proving the main theorem, we establish several lemmas on measure approximation and their $g$-functions.
    
    \begin{lemma}\label{lem_unif}
    Let us denote
    \begin{equation}\label{r_N}
        g_{\rho_N-\rho}(z):=\int_\R\log(z-s)(\rho_N(\dv s)-\rho(\dv s)).
    \end{equation}
    Let $z$ be outside a neighbourhood $U_\delta\subset\C$ of the supports of $\rho_N$ and $\rho$ such that 
    \begin{equation}
    \inf\{|z-s|:z\not\in U_\delta,s\in\mathrm{supp} (\rho_N)\cup\mathrm{supp}(\rho)\}\ge\delta>0.
    \end{equation}
    With Assumption~\ref{a1}, then for any $0<\eta<1$ there exists $C_{\eta,\delta}>0$ (given by~\eqref{C_eta}) such that
    \begin{equation}\label{unif_approx}
        |g_{\rho_N-\rho}(z)|\leq C_{\eta,\delta}\, d_{\mathrm W_\eta}(\rho_N,\rho).
    \end{equation}
    \end{lemma}

    \begin{proof}

    One only needs to show that $\log(z-s)$, as functions in $s$, satisfy the inequality for all $x,y\in\R$,
    \begin{equation}\label{log_Was}
        \left|f(x)-f(y)\right|\leq \frac{C_{\eta,\delta}}{2}\min\{|x-y|,|x-y|^\eta\}.
    \end{equation}
    Then both the real and imaginary parts of $\log(z-s)$ satisfies the same inequality. By the triangle inequality one has
    \begin{equation}
    \begin{split}
        |g_{\rho_N-\rho}(z)|&\leq \left|\int_\R\re \log(z-s)(\rho_N(\dv s)-\rho(\dv s))\right|+\left|\int_\R\im \log(z-s)(\rho_N(\dv s)-\rho(\dv s))\right|\\
        &\leq C_{\eta,\delta}\sup_{f\in\mathcal H_\eta}\left|\int f(s)(\rho_N(\dv s)-\rho(\dv s))\right|= C_{\eta,\delta} d_{\mathrm{W}_\eta}(\rho_N,\rho).
    \end{split}
    \end{equation}
    By Assumption~\ref{a1} this is uniformly small independent of $z$.

    To show~\eqref{log_Was} we write
    \begin{equation}
        |\log(z-x)-\log(z-y)|=\left|\log\left(1+\frac{y-x}{z-y}\right)\right|\le \log\left(1+\frac{|x-y|}{\delta}\right).
    \end{equation}
    For any $\eta\in(0,1]$, we take an upper bound
    \begin{equation}\label{M_eta}
        \sup_{u>0}\frac{\log(1+u)}{u^\eta}\le \frac{e^{\eta-1}}{\eta}=:M_\eta.
    \end{equation}
    To verify this, in the case $|u|\le 1$ it is
    \begin{equation}
        \frac{\log(1+u)}{u^\eta}\le u^{1-\eta}\le 1,
    \end{equation}
    and for $|u|>1$ one has
    \begin{equation}
        \frac{\log(1+u)}{u^\eta}\le \frac{\log(eu)}{u^\eta}=\frac{1+\log u}{u^\eta}\le \frac{e^{\eta-1}}{\eta}.
    \end{equation}
    Also since $\frac{e^{\eta-1}}{\eta}\ge 1$, combining the previous two bounds we claim~\eqref{M_eta}. Then one has
    \begin{equation}
        \log\left(1+\frac{|x-y|}{\delta}\right)\le M_\eta \left(\frac{|x-y|}{\delta}\right)^\eta=\frac{M_\eta}{\delta^\eta}|x-y|^\eta.
    \end{equation}
    Thus
    \begin{equation}
        |\log(z-x)-\log(z-y)|\le \min\left\{\frac{M_1}{\delta}|x-y|,\frac{M_\eta}{\delta^\eta}|x-y|^\eta\right\}
    \end{equation}
    for any $\eta\in(0,1)$. Taking 
    \begin{equation}\label{C_eta}
        C_{\eta,\delta}:=2\max\left\{\frac{M_1}{\delta},\frac{M_\eta}{\delta^\eta}\right\}
    \end{equation}
    with $M_\eta$ given in~\eqref{M_eta} we reclaims~\eqref{log_Was}. For future use, we remark that when $\delta=1$, 
    \begin{equation}\label{C_eta1}
        C_{\eta,1}=2\max\left\{1,\frac{e^{\eta-1}}{\eta}\right\}\le \frac{e^{\eta-1}}{\eta}.
    \end{equation}
    \end{proof}

    A direct consequence of Lemma~\ref{lem_unif} is the following inequality.

    \begin{lemma}\label{lem_logmuN}
        With Assumption~\ref{a1} and $|\im z|\ge\delta>0$, we have
        \begin{align}
        &\left|\int_\R\log(z-s)\rho(\dv s)\right|< \log|z|+K_{\delta,\rho},\label{bound_nu}\\
        &\left|\int_\R\log(z-s)\rho_N(\dv s)\right|<\log|z|+K_{\delta,\rho}+|g_{\rho_N-\rho}(z)|\label{bound_nuN}.
        \end{align}
        for some constant $K_{\delta,\rho}$ given by~\eqref{K_delta} independent of $z$.
    \end{lemma}
    
\begin{proof}
    We show the previous two equations by firstly showing an inequality
    \begin{equation}\label{ineq}
        |\log(z-s)|\le \log|z|+\frac{1}{2}\log(1+s^2)+\frac{1}{2}\log(1+\delta^{-2})+\pi
    \end{equation}
    for $|\im z|\ge \delta>0$.

    To see this, we note the following inequalities
    \begin{equation}
        |z-s|\le|z|\left|1-\frac{s}{z}\right|\le |z|\left(1+\frac{|s|}{|z|}\right)\le |z|\sqrt{1+s^2}\sqrt{1+\frac{1}{|z|^2}}
    \end{equation}
    where the last step is due to the Cauchy-Schwartz inequality. Since the logarithm function is increasing, we have
    \begin{equation}
        \log|z-s|\le \log|z|+\frac{1}{2}\log(1+s^2)+\frac{1}{2}\log(1+|\delta|^{-2}).
    \end{equation}
    The argument is simply bounded by $\pi$, and thus we claim~\eqref{ineq}.

    Setting
    \begin{equation}\label{K_delta}
        K_{\delta,\rho}:=\frac{1}{2}\int_\R\log(1+s^2)\rho(\dv s)+\frac{1}{2}\log(1+|\delta|^{-2})+\pi,
    \end{equation}
    one reclaims~\eqref{bound_nu}. The second bound~\eqref{bound_nuN} follows directly from~\eqref{unif_approx} and the triangular inequality.
\end{proof}

We will end this subsection with the following simple inequality, which will be used in the next subsection.

\begin{lemma}
For a given $z\in \C\setminus\R$, the following inequality holds
\begin{equation}\label{upper_bd}
    \left|\frac{\xi\log(z-s)^2}{(z-s)^{\xi+1}}\right|\le \frac{\log|\im z|^2+\pi^2}{|\im z|^{\xi+1}}
\end{equation}
for $\xi\in(0,1)$ and $s\in\R$.
\end{lemma}

\begin{proof}
    Let $z-s=re^{it}$ be its polar form. Then
    \begin{equation}
    \left|\frac{\xi\log(z-s)^2}{(z-s)^{\xi+1}}\right|\le \frac{(\log r)^2+\pi^2}{r^{\xi+1}}.
    \end{equation}
    Since $r^{\xi+1}$ decays faster than $(\log r)^2$, the right side of the previous function decays for $r>0$. Noting that $r=|z-s|\ge|\im z|$ gives the result.
\end{proof}

\subsection{Weak convergence in the classical regime}\label{s3.3}

We will show the following proposition.

\begin{proposition}
    With the notations above, in the regime where $N\to+\infty$ and $\beta_N\ll N^{-1}$, or equivalently, $c_N\to0^+$, we have $\rho_N^{(c_N)}$ converges weakly to $\rho$.
\end{proposition}

\begin{proof}
Recalling the Markov-Krein correspondence between $\rho_N^{(c_N)}$ and $\rho_N$ given by
\begin{equation}
    \int_\R\frac{\rho_N^{(c_N)}(\dv x)}{(z-x)^{c_N}}=\exp\left(-c_N\int_\R\log(z-s)\rho_N(\dv s)\right),\quad \text{for all }z\in\C\setminus\R,
\end{equation}
and recalling the notation $g_{\rho_{N}-\rho}:=\int_\R\log(z-s)(\rho_N(\dv s)-\rho(\dv s))$ given by~\eqref{r_N}, we have
\begin{align}
    \int_\R\frac{\rho_N^{(c_N)}(\dv x)}{(z-x)^{c_N}}&=\exp\left(-c_N\left(\int_\R\log(z-s)\rho(\dv s)+g_{\rho_{N}-\rho}(z)\right)\right).
\end{align}
Now we take a Taylor expansion on both sides around $c_N=0$. On the left hand side one has
\begin{equation}
    \frac{1}{(z-x)^{c_N}}=1-c_N\log(z-s)+\frac{c_N^2}{2}\frac{\log(z-s)^2}{(z-s)^\xi},\quad \text{for some }\xi\in(0,c_N).
\end{equation}
By Theorem~\ref{conv_hull}, the support of $\rho_N^{(c_N)}$is finite, and hence compact. By the dominated convergence theorem, the expansion can be interchanged with the integral, and thus we have
\begin{equation}
    \int_\R\frac{\rho_N^{(c_N)}(\dv x)}{(z-x)^{c_N}}=1-c_N\int_\R\log(z-s)\rho_N^{(c_N)}(\dv s)+\frac{c_N^2}{2}\int_\R\frac{\log(z-s)^2}{(z-s)^\xi}\rho_N^{(c_N)}(\dv s).
\end{equation}
On the right hand side an expansion of the exponential function gives
\begin{equation}
\begin{split}
    &\exp\left(-c_N\left(\int_\R\log(z-s)\rho(\dv s)+g_{\rho_{N}-\rho}(z)\right)\right)
    \\\quad&=1-c_N\left(\int_\R\log(z-s)\rho(\dv s)+g_{\rho_{N}-\rho}(z)\right)+\frac{c_N^2}{2}g_{\rho_{N}}(z)^2\exp\left(-\zeta g_{\rho_{N}}(z)
    \right),\quad \text{for some }\zeta\in(0,c_N).
\end{split}
\end{equation}
A simplification gives
\begin{equation}
    \int_\R\log(z-s)\left(\rho_N^{(c_N)}(\dv s)-\rho(\dv s)\right)=g_{\rho_{N}-\rho}(z)-\frac{c_N}{2}\left(g_{\rho_{N}}(z)^2\exp\left(-\zeta g_{\rho_{N}}(z)
    \right)-\int_\R\frac{\log(z-s)^2}{(z-s)^\xi}\rho_N^{(c_N)}(\dv s)\right).
\end{equation}

Now we take the derivative with respect to $z$ on both sides, and then interchange all integrals with the derivatives (it is not difficult to verify the convergence as all measures involved are probability measures). One obtains
\begin{equation}\label{3.39}
\begin{split}
    \int_\R\frac{\rho_N^{(c_N)}(\dv s)}{z-s}=\int_\R\frac{\rho_N(\dv s)}{z-s}&-\frac{c_N}{2}\bigg((2g_{\rho_{N}}(z)-\zeta)g'_{\rho_{N}}(z)\exp\left(-\zeta g_{\rho_{N}}(z)
    \right)
    \\&+\int_\R\frac{\xi\log(z-s)^2-2}{(z-s)^{\xi+1}}\rho_N^{(c_N)}(\dv s)\bigg).
\end{split}
\end{equation}
Since $\rho_N$ is a probability measure, one obtains the upper bound for $g'_{\rho_N}$ as
\begin{equation}
    |g'_{\rho_N}(z)|=\left|\int_\R\frac{\rho_N(\dv s)}{z-s}\right|\le \sup_{s\in\R}\frac{1}{|z-s|}\le \frac{1}{|\im z|}.
\end{equation}
Using the upper bound~\eqref{upper_bd}, we have
\begin{equation}
    \left|\int_\R\frac{\xi\log(z-s)^2-2}{(z-s)^{\xi+1}}\rho_N^{(c_N)}(\dv s)\right|\le \frac{(\log|\im z|)^2+\pi^2}{|\im z|^{\xi+1}}+\frac{2}{|\im z|^{\xi+1}}<+\infty.
\end{equation}
 Also, $|g_{\rho_N}|$ is bounded by~\eqref{bound_nuN} and ~\eqref{unif_approx}. Thus~\eqref{3.39} becomes
\begin{equation}
    \int_\R\frac{\rho_N^{(c_N)}(\dv s)}{z-s}=\int_\R\frac{\rho_N(\dv s)}{z-s}+o(1)
\end{equation}
in the $N\to+\infty$ (i.e. $c_N\to 0^+$) limit, pointwise for $z\in\C\setminus\R$.

Now we make use of the convergence theorem for Stieltjes transform (see e.g.~\cite[Theorem 2.4.4]{AGZ}): a sequence of probability measure weakly converges to another probability measure, if and only if their Stieltjes transforms converge pointwise for all $z\in\C\setminus\R$. Noting that
\begin{equation}
    \int_\R\frac{\rho_N^{(c_N)}(\dv s)}{z-s}=\int_\R\frac{\rho_N(\dv s)}{z-s}+o(1)=\int_\R\frac{\rho(\dv s)}{z-s}+o(1),\quad\text{ for all }z\in\C\setminus\R,
\end{equation}
the weak convergence $\rho_N^{(c_N)}\to\rho$ can then be concluded.
\end{proof}

\subsection{Weak convergence in the high temperature regime}\label{s3.4}

In this subsection we will show tha following proposition.

\begin{proposition}\label{prop_ht}
    With then notations above, in the regime where $\beta_N\ll N^{-1}$ and $N\to+\infty$. Equivalently, $c_N=c+o(1)$ for some $c>0$, one has $\rho_N^{(c_N)}$ converges weakly to $\rho^{(c)}$ for all $c>1$.
\end{proposition}

 By the L\'evy continuity theorem, it is equivalent to show the point-wise convergence between their characteristic functions, i.e.
\begin{equation}
    \int_\R \rho_N^{(c_N)}(\dv x)e^{itx}\to\int_\R \rho^{(c)}(\dv x)e^{itx},\quad\text{ for all }t\in\R.
\end{equation}
Since $\rho_N$ is compactly supported, by Theorem~\ref{conv_hull}, $\rho_N^{(c_N)}$ is also compactly supported. By Proposition~\ref{Fourier} we have its characteristic function being expressed by a contour integral
\begin{align}
    &\int_\R \rho_N^{(c_N)}(x)e^{it x}=\frac{\Gamma(c_N)}{(it)^{c_N-1}}\frac{1}{2\pi i}\int_{\R-i\delta\mathrm{sgn}(t)}\dv z\,\exp\left(itz-c_Ng_{\rho_N}(z)\right).    
\end{align}
We only discuss the case where $c>1$. By Proposition~\ref{Fourier} we have
\begin{align}
    &\int_\R \rho^{(c)}(x)e^{it x}=\frac{\Gamma(c)}{(it)^{c-1}}\frac{1}{2\pi i}\int_{\R-i\delta\mathrm{sgn}(t)}\dv z\,\exp\left(itz-cg_{\rho}(z)\right).
\end{align}
One aims at showing those two quantities are close as $N\to+\infty$. For convenience, let us denote
\begin{align}
    F(t)&:=\int_{\R-i\delta\mathrm{sgn}(t)}\dv z\,\exp\left(itz-cg_{\rho}(z)\right),\\
    F_N(t)&:=\int_{\R-i\delta\mathrm{sgn}(t)}\dv z\,\exp\left(itz-c_Ng_{\rho_N}(z)\right).
\end{align}
By the triangle inequality one has
\begin{align}
    \left|\int_\R \rho_N^{(c_N)}(x)e^{it x}-\int_\R \rho^{(c)}(x)e^{it x}\right|&\le\underbrace{\left|\frac{\Gamma(c_N)}{(it)^{c_N-1}}\frac{1}{2\pi i}\right|}_{I_1}\cdot\underbrace{\left|F_N(t)-F(t)\right|}_{I_2}+\underbrace{\left|\frac{\Gamma(c_N)}{(it)^{c_N-1}}-\frac{\Gamma(c)}{(it)^{c-1}}\right|\cdot\left|\frac{1}{2\pi i}F(t)\right|}_{I_3}.
\end{align}
To give an upper bound of this quantity, it requires to give bounds to each terms $I_1,I_2$ and $I_3$.
\\

One firstly states the following lemma for bounds for $I_1$ and $I_3$.
\begin{lemma}[Bounds for $I_1$ and $I_3$]\label{lem_I1I3}
    If $|c_N-c|<\varepsilon$ for some $\varepsilon>0$, then
    \begin{equation}
        I_1=\left|\frac{\Gamma(c_N)}{(it)^{c_N-1}}\frac{1}{2\pi i}\right|\le \frac{\max\{\Gamma(c-\varepsilon),\Gamma(c+\varepsilon)\}|t|^{|c+\varepsilon-1|}}{2\pi}
    \end{equation} 
    and
    \begin{equation}
        I_3=\left|\frac{\Gamma(c_N)}{(it)^{c_N-1}}-\frac{\Gamma(c)}{(it)^{c-1}}\right|\cdot\left|\frac{1}{2\pi i}F(t)\right|\le|c_N-c|B,
    \end{equation}
    where
    \begin{equation}
        B:=\left(\sup_{|w-c|\le\epsilon}|\psi(w)|+|\log|t||+\frac{\pi}{2}\right)e^{\varepsilon\left(\sup_{|w-c|\le\epsilon}|\psi(w)|+|\log|t||+\frac{\pi}{2}\right)}.
    \end{equation}
\end{lemma}
\begin{proof}
For $I_1$, the bound can be observed by notating that the Gamma function is convex on $\R^+$ and so $\Gamma(c_n)\le \max\{\Gamma(c-\varepsilon),\Gamma(c+\varepsilon)\}$, and that
\begin{equation}
    |t^{-(c_N-1)}|\le \exp|(c_N-1)\log |t||\le \exp|(c+\varepsilon-1)\log |t||.
\end{equation}

For $I_3$, we have
\begin{align}
    \left|\frac{\Gamma(c_N)}{(it)^{c_N-1}}-\frac{\Gamma(c)}{(it)^{c-1}}\right|\cdot\left|\frac{1}{2\pi i}F(t)\right|
    &\le\left|\frac{\Gamma(c_N)(it)^{c-1}}{\Gamma(c)(it)^{c_N-1}}-1\right|\cdot \left|\int_\R \rho^{(c)}(x)e^{it x}\right|\\
    &\le\left|\frac{\Gamma(c_N)}{\Gamma(c)}(it)^{c-c_N}-1\right|.
\end{align}
It remains to give an upper bound of the right side. Taking
\begin{equation}
    z:=\log \Gamma(c_N)-\log\Gamma(c)+(c-c_N)\log(it),
\end{equation}
we can use the notion of digamma function $\psi(u)=(\log\Gamma(u))'$ to obtain
\begin{equation}
    z=(c-c_N)\left(\int_0^1\psi(c+u(c-c_N))\dv u-\log(it)\right).
\end{equation}
Since $\psi$ is analytic near $c>1$, the quantity $\sup_{|w-c|\le\epsilon}|\psi(w)|<+\infty$ is bounded and one has
\begin{equation}
    |z|\le |c_N-c|\left(\sup_{|w-c|\le\epsilon}|\psi(w)|+|\log|t||+\frac{\pi}{2}\right).
\end{equation}
By using the inequalities $|e^z-1|\le |z|e^{|z|}$ and 
\begin{equation}
    \left|\frac{1}{2\pi i}F(t)\right|=\left|\int_\R e^{itx}\rho^{(c)}(\dv x)\right|\le 1
\end{equation}
we have
\begin{equation}
    \left|\frac{\Gamma(c_N)}{\Gamma(c)}(it)^{c-c_N}-1\right|\le |c_N-c|\left(\sup_{|w-c|\le\epsilon}|\psi(w)|+|\log|t||+\frac{\pi}{2}\right)e^{\varepsilon\left(\sup_{|w-c|\le\epsilon}|\psi(w)|+|\log|t||+\frac{\pi}{2}\right)}
\end{equation}
which finishes the proof.
\end{proof}

It remains to show $I_2$, the difference between $F$ and $F_N$, is small. For this purpose, using the identity
\begin{align}
    c_Ng_{\rho_N}(z)&=cg_{\rho}(z)+c_Ng_{\rho_N-\rho}(z)+(c_N-c)g_{\rho}(z)
\end{align}
we have
\begin{align}\label{F_N_ineq}
    \left|F_N(t)-F(t)\right|&=\left|\int_{\R-i\delta\mathrm{sgn}(t)}\dv z\,\underbrace{\exp\left(itz-cg_{\rho}(z)\right)}_{I_4}\underbrace{\left(\exp\left(-c_Ng_{\rho_N-\rho}(z)-(c_N-c)g_{\rho}(z)\right)-1\right)}_{I_5}\right|.
\end{align}
The rest of this subsection is to show the above quantity goes to $0$ for all $t$. For this purpose, we present some lemmas giving explicit bounds for some terms in~\eqref{F_N_ineq}.

\begin{lemma}[Bound for $I_5$]\label{lem_14}
    With the above notations, for any $\varepsilon>0$, one has
    \begin{equation}
    \left|\exp\left(-c_Ng_{\rho_N-\rho}(z)-(c_N-c)g_{\rho}(z)\right)-1\right|\leq \left(\frac{b_N}{\delta^{\varepsilon}}+\frac{|c_N-c|}{2e\varepsilon}\right)e^{b_N}(x^2+\delta^2)^{|c_N-c|/2+\varepsilon}
    \end{equation}
    with $z=x-i\delta\sgn(t)$, $x\in\R$, $\delta>0$ and $t\in\R\setminus\{0\}$, where the constant $b_N$ is given by~\eqref{b_N}.
\end{lemma}
\begin{proof}
    Using the inequality
    \begin{equation}
        |e^z-1|\leq |z|e^{|z|},\quad z\in\C,
    \end{equation}
    one has
    \begin{equation}\label{3.58}
    \begin{split}
    &\left|\exp\left(-c_Ng_{\rho_N-\rho}(z)-(c_N-c)g_{\rho}(z)\right)-1\right|\\
    &\quad\le(|c_Ng_{\rho_N-\rho}(z)|+|(c_N-c)g_{\rho}(z)|)\,\exp(|c_Ng_{\rho_N-\rho}(z)|+|(c_N-c)g_{\rho}(z)|).
    \end{split}
    \end{equation}
    By Lemma~\ref{lem_unif} and Lemma~\ref{lem_logmuN} we have
    \begin{equation}
        |g_{\rho_N-\rho}(z)|\le C_{\eta,\delta} \dv_{W_\eta}(\rho_N,\rho),\quad |g_\rho(z)|\le \log|z|+K_{\delta,\rho}.
    \end{equation}
    Substituting these two inequalities into~\eqref{3.58} we have
    \begin{equation}
    \begin{split}
    &\left|\exp\left(-c_Ng_{\rho_N-\rho}(z)-(c_N-c)g_{\rho}(z)\right)-1\right|\\
    &\quad\leq \left(c_N C_{\eta,\delta} \dv_{W_\eta}(\rho_N,\rho)+|c_N-c|(\log |z|+K_{\delta,\rho})\right)e^{c_N C_{\eta,\delta} \dv_{W_\eta}(\rho_N,\rho)+|c_N-c|K_{\delta,\rho}}\,|z|^{|c_N-c|}.
    \end{split}
    \end{equation}   
    Denote
    \begin{equation}\label{b_N}
        b_N:=c_N C_{\eta,\delta} \dv_{W_\eta}(\rho_N,\rho)+|c_N-c|K_{\delta,\rho}.
    \end{equation}
    Note that $b_N$ is independent of $z$. With $z=x-i\delta\sgn(t)$, we have $|z|=(x^2+\delta^2)^{1/2}$ and so
    \begin{equation}
    \left|\exp\left(-c_Ng_{\rho_N-\rho}(z)-(c_N-c)g_{\rho}(z)\right)-1\right|\leq \left(b_N+\frac{|c_N-c|}{2}\log(x^2+\delta^2)\right) e^{b_N}(x^2+\delta^2)^{|c_N-c|/2}.
    \end{equation}
    To simplify the previous inequality, one takes the inequalities
    \begin{equation}
        b_Ne^{b_N}(x^2+\delta^2)^{|c_N-c|/2}\le \frac{b_Ne^{b_N}}{\delta^{\varepsilon}}(x^2+\delta^2)^{|c_N-c|/2+\varepsilon},
    \end{equation}
    \begin{equation}
        \frac{|c_N-c|}{2}e^{b_N} \log(x^2+\delta^2) \le\frac{|c_N-c|}{2e\varepsilon}e^{b_N}(x^2+\delta^2)^{\varepsilon},
    \end{equation}
    which lead to the final result.
\end{proof}

\begin{proposition}[Bound for $I_4$]\label{lem_15}
Fix $\delta>0$ and $c>0$. Then for every $M>0$ there exist a constant
$A$ such that for all $x\in\mathbb R$
\begin{equation}
\exp\left(-\frac{c}{2}\int_{\mathbb R}\log((x-s)^{2}+\delta^{2})\,\rho(\mathrm{d}s)\right)
\le e^{-\frac{cA}{2}}(x^{2}+\delta^{2})^{-\tfrac{c\rho([-M,M])}{2}}.
\end{equation}
where
\begin{equation}
    A:=\rho([-M,M])\log\min\left\{\frac{1}{4},\frac{\delta^2}{\delta^2+4M^2}\right\} +\rho(\R\setminus[-M,M])\log\delta^{2}.
\end{equation}
\end{proposition}

\begin{proof}
Let us decompose
\begin{equation}
J(x):=\int_{\mathbb R}\log((x-s)^{2}+\delta^{2})\,\rho(\mathrm{d}s)
= J_{\mathrm{in}}(x)+J_{\mathrm{out}}(x),
\end{equation}
where
\begin{equation}
J_{\mathrm{in}}(x):=\int_{|s|\le M}\log((x-s)^{2}+\delta^{2})\,\rho(\mathrm{d}s), 
\qquad
J_{\mathrm{out}}(x):=\int_{|s|>M}\log((x-s)^{2}+\delta^{2})\,\rho(\mathrm{d}s).
\end{equation}

Step 1 is to show a lower bound for $J_{\mathrm{in}}(x)$. For each $|s|\le M$ and every $x\in\mathbb R$, we will check that
\begin{equation}\label{3.65}
(x-s)^{2}+\delta^{2}\ge\min\left\{\frac{1}{4},\frac{\delta^2}{\delta^2+4M^2}\right\}(x^2+\delta^2).
\end{equation}
Let us show~\eqref{3.65} first. In the case $|x|<2M$, one has
\begin{equation}\label{ineq1}
    (x-s)^{2}+\delta^{2}\ge\delta^2\ge \frac{\delta^2}{\delta^2+4M^2}(x^2+\delta^2).
\end{equation}
In the case $|x|\ge 2M$, one has $|x|-M\ge \frac{x}{2}$ and so a triangualar inequality gives
\begin{equation}
    (x-s)^2\ge ||x|-M|^2\ge \frac{x^2}{4}.
\end{equation}
This implies
\begin{equation}\label{ineq2}
    (x-s)^{2}+\delta^{2}\ge\delta^{2}+\frac{x^2}{4}\ge\frac{1}{4}(\delta^2+x^2).
\end{equation}
Combining the two inequalities~\eqref{ineq1} and~\eqref{ineq2} we find the constant $\min\left\{\frac{1}{4},\frac{\delta^2}{\delta^2+4M^2}\right\}$. Taking logarithms and integrating over $\rho$ for~\eqref{3.65}
gives
\begin{equation}\label{Jin}
    J_{\mathrm{in}}(x)\ge \rho([-M,M])\log(x^2+\delta^2)+\rho([-M,M])\log\min\left\{\frac{1}{4},\frac{\delta^2}{\delta^2+4M^2}\right\}.
\end{equation}

Next we will show a lower bound for $J_{\mathrm{out}}(x)$. Trivially, $(x-s)^{2}+\delta^{2}\ge \delta^{2}$ for all $s,x$, hence
\begin{equation}\label{Jout}
J_{\mathrm{out}}(x)\ge\rho(\R\setminus[-M,M])\log\delta^{2}.
\end{equation}
Combining~\eqref{Jin} and~\eqref{Jout} gives
\begin{equation}
J(x)\ge \rho([-M,M])\log(x^2+\delta^2)+\rho([-M,M])\log\min\left\{\frac{1}{4},\frac{\delta^2}{\delta^2+4M^2}\right\} +\rho(\R\setminus[-M,M])\log\delta^{2}.
\end{equation}
Exponentiating this gives the desired result.
\end{proof}

\begin{lemma}[Bound for $I_2$]\label{lem_I2}
    We take $\delta=1$ and $c>1$. For any $\varepsilon\in(0,\frac{c-1}{5})$, $M>0$ with $\rho([-M,M])\in(\frac{1+5\varepsilon}{c},1]$ and $N>N_0$ (such that $|c_N-c|<\varepsilon$), one has
    \begin{equation}
    \left|F_N(t)-F(t)\right| \leq 2\sqrt{\frac{\pi}{\varepsilon}}e^{\frac{c}{2}\log(1+4M^2)+|t|}\left(b_N+\frac{|c_N-c|}{2e\varepsilon}\right)e^{b_N}
\end{equation}
where
\begin{equation}
    b_N\le \frac{(c+\varepsilon) e^{\eta-1}}{\eta} \dv_{W_\eta}(\rho_N,\rho)+|c_N-c|\left(\frac{1}{2}\int_\R\log(1+s^2)\rho(\dv s)+4\right),\quad \eta\in(0,1).
\end{equation}
\end{lemma}

\begin{proof}

Lemma~\eqref{lem_15} shows that for $z\in\R-i\delta\sgn(t)$
\begin{equation}\label{3.73}
    \left|\exp\left(itz-cg_{\rho}(z)\right)\right|\le e^{-\frac{cA}{2}}(x^2+\delta^2)^{-c\rho([-M,M])/2}e^{|t|\delta}.
\end{equation}
Substituting the inequality in Lemma~\ref{lem_14} and~\eqref{3.73} into~\eqref{F_N_ineq} gives
\begin{equation}
\begin{split}
    \left|F_N(t)-F(t)\right|& \leq e^{-\frac{cA}{2}+|t|\delta}\left(\frac{b_N}{\delta^{\varepsilon}}+\frac{|c_N-c|}{2e\varepsilon}\right)e^{b_N}\int_\R\dv x\,(x^2+\delta^2)^{-c\rho([-M,M])/2+|c_N-c|/2+\varepsilon}.
\end{split}
\end{equation}
As $c>1$, one takes $\varepsilon\in(0,\frac{c-1}{5})$ and $M>0$ with $\rho([-M,M])\in(\frac{1+5\varepsilon}{c},1]$. Then there exists $N_0$ such that for all $N>N_0$ one has  $|c_N-c|<\varepsilon$. Then the inequality
\begin{equation}
    \alpha_N:=c\rho([-M,M])/2-|c_N-c|/2-\varepsilon\in\left(\frac{1}{2}+\varepsilon,\frac{c}{2}\right).
\end{equation}
holds. So the previous integral is absolute integrable and can be computed as
\begin{equation}
\begin{split}
    \int_\R\dv x\,(x^2+\delta^2)^{-\alpha}=\frac{\sqrt{\pi}\Gamma(\alpha_N-\frac{1}{2})}{\Gamma(\alpha_N)}\delta^{1-2\alpha_N}.
\end{split}
\end{equation}
Note the two inequalities simplify the bound
\begin{equation}
    \left|\delta^{1-2\alpha_N}\right|\le 1+\delta^{1-c},\qquad
    \left|\frac{\sqrt{\pi}\Gamma(\alpha_N-\frac{1}{2})}{\Gamma(\alpha_N)}\right|\le\sqrt{\frac{\pi}{\alpha_N-1/2}}\le\sqrt{\frac{\pi}{\varepsilon}}.
\end{equation}
We thus have
\begin{equation}
    \left|F_N(t)-F(t)\right| \leq \sqrt{\frac{\pi}{\varepsilon}}(1+\delta^{1-c})e^{-\frac{cA}{2}+|t|\delta}\left(\frac{b_N}{\delta^{\varepsilon}}+\frac{|c_N-c|}{2e\varepsilon}\right)e^{b_N}.
\end{equation}
To further simplify this bound, we take $\delta=1$ (by the arbitrariness of $\delta$) and $M>\delta$. Then
\begin{equation}
    A\ge -\log(1+4M^2)
\end{equation}
and one thus has
\begin{equation}
    \left|F_N(t)-F(t)\right| \leq 2\sqrt{\frac{\pi}{\varepsilon}}e^{\frac{c}{2}\log(1+4M^2)+|t|}\left(b_N+\frac{|c_N-c|}{2e\varepsilon}\right)e^{b_N},
\end{equation}
where
\begin{equation}
    b_N\le \frac{(c+\varepsilon) e^{\eta-1}}{\eta} \dv_{W_\eta}(\rho_N,\rho)+|c_N-c|\left(\frac{1}{2}\int_\R\log(1+s^2)\rho(\dv s)+4\right),\quad \eta\in(0,1).
\end{equation}
(we use the inequality $\frac{1}{2}\log 2+\pi<4$ only for convenience). 
\end{proof}

Now we can prove Proposition~\ref{prop_ht}.
\begin{proof}[Proof of Proposition~\ref{prop_ht}] 
    Theorem~\ref{Fourier} tells us that $\delta$ can be arbitrary. We take $\delta=1$ and $c>1$. Summarising Lemma~\eqref{lem_I1I3} and Lemma~\eqref{lem_I2} we see that for any $\varepsilon\in(0,\frac{c-1}{5})$, $M_\varepsilon>0$ with $\rho([-M_\varepsilon,M_\varepsilon])\in(\frac{1+5\varepsilon}{c},1]$ and $N>N_0$ one has
    \begin{equation}\label{3.99}
    \begin{split}
        &\left|\int_\R \rho_N^{(c_N)}(x)e^{it x}-\int_\R \rho^{(c)}(x)e^{it x}\right|\le \tilde C\left(\tilde b_N+\frac{|c_N-c|}{2e\varepsilon}\right)e^{\tilde b_N}+|c_N-c|B,
    \end{split}
    \end{equation}
    where
    \begin{equation}
        \tilde C=\frac{\max\{\Gamma(c-\varepsilon),\Gamma(c+\varepsilon)\}|t|^{c+\varepsilon-1}}{2\pi}\cdot 2\sqrt{\frac{\pi}{\varepsilon}}e^{\frac{c}{2}\log(1+4M_\varepsilon^2)+|t|},
    \end{equation}
    \begin{equation}\label{3.101}
        \tilde b_N:=\frac{(c+\varepsilon) e^{\eta-1}}{\eta} \dv_{W_\eta}(\rho_N,\rho)+|c_N-c|\left(\frac{1}{2}\int_\R\log(1+s^2)\rho(\dv s)+4\right),
    \end{equation}
    \begin{equation}
        B=\left(\sup_{|w-c|\le\epsilon}|\psi(w)|+|\log|t||+\frac{\pi}{2}\right)e^{\varepsilon\left(\sup_{|w-c|\le\epsilon}|\psi(w)|+|\log|t||+\frac{\pi}{2}\right)}.
    \end{equation}
    The only $N$-dependent terms are $\tilde b_N$ and $|c_N-c|$, both of which are small when $N\to+\infty$. Thus we have shown the Fourier transforms of $\rho_N^{(c_N)}$ and $\rho^{(c)}$ are close to each other, which by L\'evy continuity implies the weak convergence.
\end{proof}

\subsection{Proofs for Theorem~\ref{thm_main} and~\ref{thm_main_m}}

Let us firstly show~\eqref{1.20}. By Corollary~\ref{cor} we have 
\begin{equation}
    \mathcal B_{\vec a^{(N)}}\left(u;N,\frac{\beta_N}{2}\right)=\int_\R e^{ux}\rho_N^{(c_N)}(\dv x),
\end{equation}
with $c_N=N\beta_N/2$, where $\rho_N$ is given by
\begin{equation}
    \rho_N(\dv x)=\frac{1}{N}\sum_{j=1}^N\delta_{a_j^{(N)}}(\dv x)
\end{equation}
and $\rho_N^{(c_N)}$ is the distribution of the random mean of the Dirichlet process with parameter measure $c_N\rho$. By Section~\ref{s3.3}, in the limit where $c_N:=N\beta_N/2\to 0$ as $N\to+\infty$, the measure $\rho_N^{(c_N)}$ converges weakly to $\rho$. Thus, by the dominated convergence theorem,
\begin{equation}
    \lim_{N\to+\infty}\int_\R e^{ux}\rho_N^{(c_N)}(\dv x)=\int_\R e^{ux}\rho(\dv x),
\end{equation}
which finishes the proof.

Now for~\eqref{1.24}, its proof is almost identical, except that in the $N\to+\infty$ one has instead $c_N:=N\beta_N/2\to c>0$. Then by Section~\ref{s3.4}, one has the weak convergence $\rho_N^{(c_N)}\to\rho^{(c)}$. Then applying the dominated convergence theorem gives us the result.

For Equations~\eqref{1.22} and~\eqref{1.26}, we note that by Corollary~\ref{cor},
\begin{equation}
    \mathcal F_{\log\vec a^{(N)}}\left(u;N,\frac{\beta_N}{2}\right)=\int_0^{+\infty}x^u\rho_N^{(c_N)}(\dv x).
\end{equation}
By the weak convergences to $\rho$ in the classical regime and to $\rho^{(c)}$ in the high-temperature regime, it remains to use the dominated convergence theorem to conclude that
\begin{equation}
    \displaystyle\int_0^{+\infty}x^u\rho_{N}^{(c_N)}(\dv x)\to
    \begin{cases}
    \displaystyle\int_0^{+\infty}x^u\rho(\dv x),\text{ in the classical regime,}
    \vspace{.5em}
    \\
    \displaystyle\int_0^{+\infty}x^u\rho^{(c)}(\dv x),\text{ in the high temperature regime.}
    \end{cases}
\end{equation}
In the high-temperature case, we require in addition $\re u\ge 1$ such that the absolute integrability of both sides is given by Corollary~\ref{cor_tail}. The arguments will be similar to the multivariate Bessel function case.

\subsection{Proof for Theorem~\ref{thm_3}}

Theorem~\ref{thm_3} is only a simple corollary of~\eqref{3.99}. Now since $a_1^{(N)},\ldots,a_N^{(N)}$ are random points, we have that $\rho^{(c_N)}_N(\dv x)$ is a random measure. Then extended $\eta$-Wasserstein distance $\dv_{W_\eta}(\rho_N,\rho)$ is then a random variable. For~\eqref{3.99}, both sides are random variables. One has
\begin{equation}\label{3.108}
    \E\left|\int_\R e^{itx}(\rho_N^{(c_N)}-\rho^{(c)})(\dv x)\right|\le \E\left[\tilde C\left(\tilde b_N+\frac{|c_N-c|}{2e\varepsilon}\right)e^{\tilde b_N}+|c_N-c|B\right].
\end{equation}
On the right hand side, the randomness only comes from $\dv_{W_\eta}(\rho_N,\rho)$ in $\tilde b_N$ which by~\eqref{3.101} can be rewritten as
\begin{equation}
    \tilde b_N=\lambda \dv_{W_\eta}(\rho_N,\rho)+|c_N-c|\kappa
\end{equation}
and both $\lambda$ and $\kappa$ are deterministic. Therefore, with $|c_N-c|\to 0$, the convergence of the expectation on the right side of~\eqref{3.108} boils down to the convergence of
\begin{equation}
    \E \left[\dv_{W_\eta}(\rho_N,\rho)\exp \dv_{W_\eta}(\rho_N,\rho)\right]\to 0,\qquad \sup_n\E \left[\exp \dv_{W_\eta}(\rho_N,\rho)\right]<+\infty,
\end{equation}
which is implied directly by Assumption~\ref{a2}.

\section{Discussions}\label{s4}

\subsection{Comparison with results in \texorpdfstring{\cite{BGCG22}}{BGCG22}}

As mentioned in the introduction section, Benaych-Georges, Cuenca and Gorin~\cite{BGCG22} develop the notion of \textit{Bessel generating function} $G_{N,\theta}(\vec x;\mu)$ of a symmetric probability measure $\mu$ on $\R^N$, which is an integral transform with the multivariate Bessel function as the kernel:
\begin{equation}
    G_{N,\theta}(\vec x;\mu):=\E\mathcal B_{\vec a}(\vec u;\theta)=\int_{\R^N}\mathcal B_{\vec a}(\vec u;\theta)\mu(\dv \vec a).
\end{equation}
They define the multivariate $\beta$-convolution of two measures $\mu_1$ and $\mu_2$ by multiplying their Bessel generating functions. The multivariate high temperature convolution will be the high temperature limit $\beta N\to c>0$ of the previous convolution.
	
They consider a sequence of measures satisfying particular limiting assumptions, such that law of large numbers can be described. It is stated that there is a relation between the limiting one-dimensional marginal distribution of $\mu_N$ and the high temperature limit of the rank-one Bessel generating function of $\mu_N$. This relation is described by using moments and $\gamma$-cumulants.

\begin{theorem}[~\cite{BGCG22}]\label{BGCG}
    Let $a_1^{(N)},\ldots,a_N^{(N)}\in\R$ be random and has a (uniformly in $n$) compactly supported empirical distribution
    \begin{equation}
        \rho_N(\dv x):=\frac{1}{N}\sum_{j=1}^N\delta_{a_j^{(N)}}(\dv x).
    \end{equation}
    If they are $c$-LLN–appropriate (for details check~\cite{BGCG22}) and $\rho_N\to\rho$ in moments, under suitable exponential tail control, one has
    \begin{equation}
        \lim_{N\to+\infty,\,\theta N\to c}G(\underbrace{u,0,\ldots,0}_{N\text{ entries}};\mu_N)=\int_{\R}e^{ux}\rho^{(c)}(\dv x),\quad u\in\C,
    \end{equation}
    where $\rho^{(c)}$ and $\rho$ satisfies a Markov-Krein correspondence. 
\end{theorem}

\begin{remark}
    In the case $c>1$, the previous theorem is comparable to our results (Theorem~\ref{thm_3}), though there are several difference.
    \begin{enumerate}
        \item Theorem~\ref{BGCG} makes the assumption that every measure encountered has compact support, and thus moments and cumulants can be defined without a problem, while we do not admit such an assumption.
        
        \vspace{.5em}
        \item Also, the convergence $\rho_N\to\rho$ together with the requirement of being $c$-LLN-appropriate are comparable to our Assumption~\ref{a2}.
        
        \vspace{.5em}
        \item Theorem~\ref{BGCG} considers a general $u\in\C$, while we only have $u\in i\R$.
    \end{enumerate} 
\end{remark}

\subsection{\texorpdfstring{$c$-}-cumulants}

In the proofs in~\cite{BGCG22}, a notion of $c$-cumulants has been introduced (in their language it is the $\gamma$-cumulant, where their $\gamma$ is the same as our $c$) for distributions with finite moments. Let $\rho$ be a probability measure that has finite moments $\{m_k\}$. The $c$-cumulants of $\rho$ is a sequence $\{\kappa_l\}$ defined by some specific linear combinations of $m_k$.

Here we give an alternative definition of the $c$-cumulants. Noting that if we set
\begin{equation}\label{kappa}
    \tilde\kappa_l=(l-1)!\kappa_l
\end{equation}
and if there is another measure $\mu$ whose (classical) cumulants match the sequence $\{\tilde\kappa_l\}$, then $\mu$ and $\rho$ satisfy the Markov-Krein correspondence,
\begin{equation}
    \int_\R\frac{\mu(\dv x)}{(z-x)^c}=\exp\left(-c\int_\R\log(z-x)\rho(\dv x)\right).
\end{equation}
From Section~\ref{s2} we know that $\mu$ exists and it is uniquely given by $\rho^{(c)}$, the distribution of the random mean of the Dirichlet process with parameter measure $c\rho$. The paper~\cite{BGCG22} shows that moments and $c$-cumulants of $\rho$ are in a one-to-one correspondence. We claim that it can also be implied by our results. Firstly we state the following proposition regarding existence of moments.

\begin{proposition}
    Let $\rho$ and $\rho^{(c)}$ be specified above. Then $\rho$ has finite absolute moments if and only if $\rho^{(c)}$ has finite absolute moments.
\end{proposition}

\begin{proof}
    For the if part, We multiply $z^c$ on the both sides to obtain
    \begin{equation}
        \int_\R\frac{z^c\rho^{(c)}(\dv x)}{(z-x)^c}=\exp\left(-c\int_\R(\log(z-x)-\log(z))\rho(\dv x)\right).
    \end{equation}
    Since $z^c(z-x)^{-c}$ is analytic in $z$ at $\infty$, if $\rho^{(c)}$ has finite absolute moments, by the dominated convergence theorem the expansion can be swapped with the integral, and so the left side has a series expansion in $z$ at $\infty$. Thus both sides are analytic at infinity and so is $\int_\R(\log(z-x)-\log(z))\rho(\dv x)$, the expansion of which gives moments of $\rho$.

    The only if direction can be implied by Corollary~\ref{cor_tail}.
\end{proof}

In the case either (hence both) $\rho$ or $\rho^{(c)}$ has all absolute moments, their moments satisfy the Markov-Krein correspondence, and so one can say that the relate $c$-cumulants, cumulants of $\rho^{(c)}$ and moments of $\rho$ by the following proposition.

\begin{proposition}[\cite{BGCG22} for the case $\rho$ being compactly supported]
    Let $\rho$ and $\rho^{(c)}$ be specified above and $\rho$ has finite absolute moments. Then $c$-cumulants of $\rho$ and cumulants of $\rho^{(c)}$ are related by~\eqref{kappa}. The moments of $\rho^{(c)}$ and $\rho$ are related by an expansion in $z$ around $\infty$ of the Markov-Krein correspondence.
\end{proposition}

\section*{Acknowledgements}
We thanks discussions with Tianshu Cong, Cesear Cuenca, Joon Lee, Pierre Mergny and Jiaming Xu and the proofreading of the first draft by Joon Lee and Pierre Mergny. JZ acknowledges financial support of FWO Odysseus grant No. G0DDD23N and FWO Fellowship Junior 1234325N.

\appendix

\section{Proofs in Section~\ref{s2}}

\subsection{Proofs of Theorem~\ref{equiv} and Theorem~\ref{Exs_unq}}\label{A1}

Although the arguments are already contained in~\cite{LR04}, we provide them in line with our notions, only for the sake of completeness.

To do so, let us firstly consider the simple case when $\rho$ is a discrete measure, finitely supported on $\{x_1,\ldots,x_n\}\subset\R$, given by
\begin{equation}\label{discrete_measure}
    \rho(\dv x)=\sum_{j=1}^n a_j\delta_{x_j}(\dv x)
\end{equation}
where the weights satisfy $a_1,\ldots,a_n> 0$ and $a_1+\ldots+a_n=1$. It is not difficult to verify the Lauricella formula
\begin{align}\label{Dir1}
    \int_{\R^n}\frac{\mathrm{Dir}_{(ca_1,\ldots,ca_n)}(\dv \sigma_1,\ldots,\dv \sigma_n)}{(z-\sum_{j=1}^nx_j\sigma_j)^c}&=\prod_{j=1}^n\frac{1}{(z-x_j)^{a_j}}=\exp\left(-c\int_\R\log(z-s)\rho(\dv s)\right).
\end{align}

This identity leads us to the Markov-Krein correspondence between $\rho^{(c)}$ and $\rho$ when $\rho$ is finitely supported. We consider the random measure $D_{c\rho}$. By definition of the Dirichlet measure, the random vector 
\begin{equation}
(D_{c\rho}(\{x_1\}),\ldots,D_{c\rho}(\{x_n\}),D_{c\rho}(\R\setminus\{x_1,\ldots,x_n\}))
\end{equation}
follows a Dirichlet distribution $\mathrm{Dir}_{(ca_1,\ldots,ca_n,0)}$. By convention of the Dirichlet distribution, this also implies its first $n$ entries follows a Dirichlet distribution $\mathrm{Dir}_{(ca_1,\ldots,ca_n)}$ while the last entry is $D_{c\rho}(\omega)(\R\setminus\{x_1,\ldots,x_n\})=0$ almost surely. Hence
\begin{align}\label{Dir2}
    \int_{\R^n}\varphi\left(\sum_{j=1}^nx_j\sigma_j\right)\mathrm{Dir}_{(ca_1,\ldots,ca_n)}(\dv \sigma_1,\ldots,\dv \sigma_n)&=\E\left[\varphi\left(\int_{\{x_1,\ldots,x_n\}} x D_{c\rho}(\dv x)\right)\right]
\end{align}
holds for any bounded measurable function $\varphi$. Combining~\eqref{Dir1} and~\eqref{Dir2} we have
\begin{equation}
    \E\left[\left(z-\int_{\{x_1,\ldots,x_n\}} x D_{c\rho}(\dv x)\right)^{-c}\right]=\exp\left(-c\int_\R\log(z-s)\rho(\dv s)\right),
\end{equation}
which holds for all $z\in\C\setminus\R$. 

Let $\{\xi_n(x)\}_{n=1,2,\ldots}$ be a sequence of simple functions on $\R$, such that for $x\geq 0$, $\xi_n(x)$ approaches $x$ from below, while for $x<0$, $\xi_n(x)$ approaches $x$ from above, and $|\xi_n(x)|$ approaches $|x|$ from below for all $x$. This function has an image 
\begin{equation}
    \xi_n(\R)=\left\{y_j:j\in J_n\right\},
\end{equation}
which contains only finitely many different values. We also denote the preimage $\xi_n^{-1}(y_j)$ for some $j\in J_n$, which are measurable sets forming a partition of $\R$, i.e. $\R=\bigcup_{j\in J_n}\xi_n^{-1}(y_j)$.

Then we can define the push-forward measure $\xi_{n,\sharp}\rho$ (under the mapping $\xi_n$) by
\begin{equation}
    \int_B\xi_{n,\sharp}\rho(\dv x):=\int_\R \ind_{\xi_n(x)\in B}\,\rho(\dv x),\text{ for all Borel set }B.
\end{equation}
Recalling the change-of-variable property for the push-forward measure, we have for any measurable function $\varphi$
\begin{equation}
    \int_\R \varphi(x)\xi_{n,\sharp}\rho(\dv x)=\int_\R \varphi(\xi_n(x))\rho(\dv x),
\end{equation}
provided that either side is integrable. If we furthermore let $\varphi$ to be bounded, taking $n\to+\infty$ on both sides yields $\xi_{n,\sharp}\rho$ converges weakly to $\rho$ by the dominated convergence theorem.

\begin{lemma}[Eq. (3.1) in~\cite{LR04}]
    Let $\varphi(x)$ be a measurable function with $|\varphi(x)|\leq C|x|$ for some constant $C>0$. With the notations above and the assumption~\eqref{cond_nu},
    \begin{align}\label{Lauricella1}
        \E\left[\left(z-\int_\R\varphi(\xi_n(x))D_{c\rho}(\dv x)\right)^{-c}\right]=\exp\left(-c\int_\R\log(z-\varphi(s))\xi_{n,\sharp}\rho(\dv s)\right).
    \end{align}
\end{lemma}
\begin{proof}
    We take the partition $\R=\bigcup_{j\in J_n}\xi_n^{-1}(y_j)$, and we denote the value $a_j:=\rho(\xi_n^{-1}(y_j))$. The random vector
    \begin{equation}
    (\sigma_j)_{j\in J_n}=(D_{c\rho}(\xi_n^{-1}(y_j)))_{j\in J_n}
    \end{equation}
    follows a Dirichlet distribution $\mathrm{Dir}_{(ca_j)_{j\in J_n}}$. Then by definition of the Dirichlet process, one has
    \begin{equation}
        \int_{\R} \varphi(\xi_n(x)) D_{c\rho}(\dv x)=\sum_{j\in J_n} \varphi(y_j)\sigma_j.
    \end{equation}
    The $j$-sum is finite and hence there is no convergence issue. Then by the Lauricella formula~\eqref{Dir1} we have
    \begin{align}\label{2.18}
        \E\left[\left(z-\int_{\R} \varphi(\xi_n(x)) D_{c\rho}(\dv x)\right)^{-c}\right]
        &=\int_{\R^n}\frac{\mathrm{Dir}_{(ca_j)_{j\in J_n}}(\dv \vec \sigma)}{(z-\sum_{j\in J_n} \varphi(y_j)\sigma_j)^c}=\prod_{j\in J_n}\frac{1}{(z-\varphi(y_j))^{ca_j}}.
    \end{align}
    Since $a_j=\rho(\xi_n^{-1}(y_j))$, one has
    \begin{equation}\label{2.19}
    \begin{split}
        \prod_{j\in J_n}\frac{1}{(z-\varphi(y_j))^{ca_j}}&=\exp\left(-c\sum_{j\in J_n}\log(z-\varphi(y_j))\rho(\xi_n^{-1}(y_j))\right)
        \\&=\exp\left(-c\int_\R\log(z-\varphi(\xi_n(s)))\rho(\dv s)\right)
        =\exp\left(-c\int_\R\log(z-\varphi(s))\xi_{n,\sharp}\rho(\dv s)\right).
    \end{split}
    \end{equation}
    Combining~\eqref{2.18} and~\eqref{2.19} finishes the proof.
\end{proof}

Now we are in a good shape to prove Theorem~\ref{equiv} and Theorem~\ref{Exs_unq}.

\begin{proof}[Proof of Theorem~\ref{equiv}] Using~\eqref{Lauricella1} and taking $z=1$ and $\varphi(x)=-|x|$, one obtains
\begin{align}
    \E\left[\left(1+\int_\R|\xi_n(x)|D_{c\rho}(\dv x)\right)^{-c}\right]=\exp\left(-c\int_\R\log(1+|s|)\xi_{n,\sharp}\rho(\dv s)\right).
\end{align}
Assuming~\eqref{cond_nu}, one takes the limit $n\to+\infty$ on both sides. On the left hand side the integrand is bounded by the assumption, while on the right hand side the limit can be interchanged with the integral because of the weak convergence. Thus by the dominated convergence theorem we obtain
\begin{align}\label{2.25}
    \E\left[\left(1+\lim_{n\to+\infty}\int_{\R} |\xi_n(x)| D_{c\rho}(\dv x)\right)^{-c}\right]
    &=\exp\left(-c\int_\R\log(1+|s|)\rho(\dv s)\right).
\end{align}
Since the right side is none-zero, it implies that the left side is none-zero, and so the limit $\lim_{n\to+\infty}\int_{\R} |\xi_n(x)| D_{c\rho}(\dv x)$ is finite almost surely. 

On the other hand, assuming $\int_{\R} |x| D_{c\rho}(\dv x)$ is finite almost surely, then by a similar argument the right side of~\eqref{2.25} is non-zero, and thus we have~\eqref{cond_nu}.
\end{proof}

\begin{proof}[Proof of Theorem~\ref{Exs_unq}]
Here we will prove a slightly more general result: for all $z,u\in\C$ such that $z-us\not\in\R$ for all $s\in\R$, we have
\begin{equation}\label{MKR'}
    \int_\R\frac{\rho^{(c)}(\dv x)}{(z-ux)^c}=\exp\left(-c\int_\R\log(z-us)\rho(\dv s)\right).
\end{equation}

For the existence of $\rho^{(c)}$ let us again start with~\eqref{Lauricella1} and take $\varphi(x)=ux$. Now one takes the limit $n\to+\infty$ on both sides, and again by the dominated convergence theorem we obtain
\begin{align}
    \E\left[\left(z-u\lim_{n\to+\infty}\int_{\R} \xi_n(x) D_{c\rho}(\dv x)\right)^{-c}\right]
    &=\exp\left(-c\int_\R\log(z-us)\rho(\dv s)\right).
\end{align}
The limit $n\to+\infty$ can be again interchanged with the $x$-integral because of the existence of the random mean.

The uniqueness can be checked without too much effort using Theorem~\ref{Fourier}. Since each $\rho$ gives a Fourier transform of $\rho^{(c)}$, by the Fourier inversion formula this uniquely defines $\rho^{(c)}$.
\end{proof}

\subsection{Proof of Theorem~\ref{Fourier}}\label{A:Fourier}


We assume that $\rho^{(c)}$ is a probability measure satisfying~\eqref{cond_nu}, i.e. $\rho\in\mathcal V$. For any $M>0$ we define a truncated version of $\rho$ as
\begin{equation}
    \rho^{[M]}(\dv s):=\rho((-\infty,-M])\delta_{-M}(\dv s)+\ind_{(-M,M)}(s)\rho(\dv s)+\rho([M,+\infty))\delta_{M}(\dv s).
\end{equation}
Now $\rho^{[M]}$ is also a probability measure and its support is $[-M,M]$. One can also see that $\rho^{[M]}\to \rho$ weakly as $M\to\infty$.

The following Lemma has been used implicitly in~\cite[Section 7]{LR04} without giving a proof. It is a special case of Corollary~\ref{weak_conv0} (which is based on results in~\cite{DelloSchiavo19}), while here we provide a alternative elementary proof which does not rely on Corollary~\ref{weak_conv0} and results in~\cite{DelloSchiavo19}.

\begin{lemma}\label{weak_conv}
    Let $\rho^{[M]}$ and $\rho$ be stated above. For $c>0$, let $(\rho^{[M]})^{(c)}$ be the random mean of $D_{c\rho^{[M]}}$, and $\rho^{(c)}$ be the random mean of $D_{c\rho}$. Then $(\rho^{[M]})^{(c)}\to\rho^{(c)}$ weakly as $M\to\infty$.
\end{lemma}
\begin{proof}
    For a given $\rho^{[M]}$, let us take the sequence of simple functions $\{\xi_n\}$ by
    \begin{equation}
        \xi_n(x)=\begin{cases}
        \lfloor nx\rfloor/n,&x\in[0,n)\setminus\{ M\},\\
        M,&x=M,\\
        n,&x>n,\\
        -\xi_n(-x),& x<0,
        \end{cases}\qquad n> M.
    \end{equation}
    Here without loss of generality we assume $n>M$. It is straight forward to check that $\xi_n(x)$ approaches $x$ from above on the positive axis and from below on the negative axis, and $|\xi_n(x)|\to |x|$ from below for all $x$. The image of this function is
    \begin{equation}
        \xi_n(\R)=\{p/n:p=-n^2,\ldots,-1,0,1,\ldots,n^2\}\cup \{\pm M\}.
    \end{equation}
    
   Let $\varphi$ be a continuous function with $|\varphi|\leq 1$. Then
    \begin{align}
        \E &\left[\varphi\left(\int \xi_n(x)D_{c\rho^{[M]}}(\dv x)\right)\right]=\int\varphi\left(-M\sigma_-+\left(\sum_{\substack{p\in\Z\cap(-nM,nM)}} \frac{p}{n}\sigma_p\right)+M\sigma_+\right){\mathrm{Dir}_{(c\vec a)}(\dv \vec \sigma)}
    \end{align}
    where $\vec \sigma=(\sigma_-,\sigma_{\lceil -nM\rceil},\ldots,\sigma_{\lfloor nM\rfloor},\sigma_+)$ and $\vec a=(a_-,a_{\lceil -nM\rceil},\ldots,a_{\lfloor nM\rfloor},a_+)$ with 
    \begin{align}
        &a_p=\rho^{[M]}(\xi_n^{-1}(\{p/n\})),\quad p\in\Z\cap(-nM,nM),\\
        &a_-=\rho^{[M]}(\{-M\}),\\
        &a_+=\rho^{[M]}(\{M\}).
    \end{align}
    Since $\xi_n^{-1}(\{p/n\})$ in contained in $(-M,M)$, we have
    \begin{align}
        &a_p=\rho(\xi_n^{-1}(\{p/n\})),\quad p\in\Z\cap(-nM,nM),
        \\
        &a_-=\rho((-\infty,-M]),\\
        &a_+=\rho([M,+\infty)).
    \end{align}
    Thus by the definition of the Dirichlet process, one has
    \begin{align}
        \E &\left[\varphi\left(\int \xi_n(x)P_{c\rho^{[M]}}(\dv x)\right)\right]=\E \left[\varphi\left(\int \left(-M\ind_{x\leq -M}+\xi_n(x)\ind_{x\in(-M,M)}+M\ind_{x\geq -M}\right)D_{c\rho}(\dv x)\right)\right].
    \end{align}
    We then take the $N\to+\infty$ limit on both sides. On the left hand side, because of the continuity of $\varphi$ and the boundedness of $\xi_n(x)$ by $|x|$, one can push the limit inside both integrals and obtain
    \begin{equation}
        \lim_{n\to+\infty}\E \left[\varphi\left(\int \xi_n(x)D_{c\rho^{[M]}}(\dv x)\right)\right]=\E \left[\varphi\left(\int xD_{c\rho^{[M]}}(\dv x)\right)\right].
    \end{equation}
    On the right side by the continuity of $\varphi$ and the boundedness of the integrand, the $n$-limit can be pushed in to replace $\xi_n(x)$ by $x$. One thus obtains
    \begin{equation}
        \E \left[\varphi\left(\int xD_{c\rho^{[M]}}(\dv x)\right)\right]=\E \left[\varphi\left(\int \left(-M\ind_{x\leq -M}+x\ind_{x\in(-M,M)}+M\ind_{x\geq -M}\right)D_{c\rho}(\dv x)\right)\right].
    \end{equation}
    Taking $M\to+\infty$ limit on both sides gives
    \begin{equation}
        \lim_{M\to\infty}\E \left[\varphi\left(\int xD_{c\rho^{[M]}(\dv x)}\right)\right]=\E \left[\varphi\left(\int xD_{c\rho(\dv x)}\right)\right].
    \end{equation}
    This proofs that $(\rho^{[M]})^{(c)}\to\rho^{(c)}$ weakly.
\end{proof}

We will show two lemmas that rewrite he exponential function into a contour integral. Though they are somewhat standard results in complex analysis, we provide proofs for the sake of completeness.

\begin{lemma}\label{lem_Gamma}
    For all $c>0$, $u\in\C\setminus \{0\}$ and $x\in\R$, we have
    \begin{equation}\label{gamma}
    e^{ux}={u\Gamma(c)}\frac{1}{2\pi i}\int_{\sgn(u)\mathcal C}\dv z\,\frac{e^{uz}}{(u(z-x))^{c}},
    \end{equation}
    where $\sgn$ is the sign function with complex arguments given by~\eqref{sgn}, and the contour $\mathcal C$ is the Hankel loop enclosing $(-\infty,x]$; see Fig.~\ref{contour1}.
\end{lemma}
\begin{proof}
    By the Hankel integral representation of the reciprocal Gamma function we have
    \begin{equation}\label{gamma_Hankel}
        \frac{1}{\Gamma(c)}=\frac{1}{2\pi i}\int_{\mathcal C}\dv z\,\frac{e^{z}}{z^c}.
    \end{equation}
    Here the contour $\mathcal C$ encloses the usual branch cut $(-\infty,0)$ for the function $z^{-c}$. 
    
    In the case $\re u>0$, one changes variable $z\mapsto u(z-x)$ and so the contour is also changed to $z\in x+u^{-1}\mathcal C$, which encloses the branch cut $x+u^{-1}(-\infty,0]$. When looking away from the branch point, this branch cut lies at the second ( or third) quadrant for $\re u>0$ and $\im u\le 0$ (or $\im u<0$). When $\re u>0$ and $\im u \ge 0$, one can verify that on the arc $C_R$ with radius $R$ in the second quadrant
    \begin{equation}\label{arc_bd}
    \begin{split}
        \left|\int_{C_R}\frac{e^{uz}}{(u(z-x)^c)}\right|&\le \int_{\pi/2}^{\pi }\dv\theta\frac{Re^{\re(uRe^{i\theta})}}{|u(R e^{i\theta}-x)|^c}\\
        &\le \frac{\pi}{2}\frac{R}{|u(R-|x|)|^c}\sup_{\theta\in(\pi/2,\pi)}\exp(R(\re u\cos\theta-\im u\sin\theta)).
    \end{split}
    \end{equation}
    Since $\cos\theta \le 0$ and $\sin\theta \ge 0$, the above quantity goes to $0$ as $R\to\infty$. A similar argument holds true for $\im u\le 0$, since now we are at the third quadrant and $\sin\theta\ge  0$. Thus one can apply Cauchy's theorem to deform the contour back to $\mathcal C$, which starts and ends at $-\infty$, enclosing the branch point $u(z-x)$ with a branch cut going to $-\infty$.

    In the case $\re u<0$, one can see that after the change of variable, the contour is deformed such that it starts and ends at the first (or fourth) quadrant if $\im u\ge 0$ (or $\im u<0$). A similar bound to~\eqref{arc_bd} can be performed to see that applying a Cauchy theorem the contour can be deformed such that it encloses the branch point $u(z-x)$ with a branch cut going to $+\infty$. We have
    \begin{equation}
        \frac{1}{\Gamma(c)}=u\frac{1}{2\pi i}\int_{\mathcal C}\dv z\,\frac{e^{u(z-x)}}{(u(z-x))^c}.
    \end{equation}
    A simplification gives us the result.
    
    Now let us consider the case $u=it$ and $t>0$, one changes variable $z\mapsto it(z-x)$, which deform the contour to $z\in x-i t^{-1}\mathcal C$. After applying Cauchy's theorem we then have
    \begin{equation}
        \frac{1}{\Gamma(c)}=\frac{1}{2\pi i}\int_{-i\mathcal C}it\dv z\,\frac{e^{it(z-x)}}{(it(z-x))^c}.
    \end{equation}
    The contour $i\mathcal C$ encloses the branch cut $i(-\infty,x]$.

    Similarly for the case $t<0$, the same change of variable deforms the contour from $\mathcal C$ to $-i\mathcal C$, which encloses the branch cut $-i(-\infty,x]$.
\end{proof}

\begin{lemma}\label{lem_laplace}
    For $t\ne 0$ and $\varepsilon>0$ small enough,
    \begin{equation}
        \int_{I_t} \exp(R(\varepsilon \cos\theta-t \sin\theta))\dv \theta=O\left(\frac{e^{\varepsilon R}}{t R}\right),
    \end{equation}
    where $I_t=[\pi/2,\pi]$ if $t>0$, and $I_t=[\pi,3\pi/2]$ if $t<0$.
\end{lemma}

\begin{proof}
    The proof is a simple application of Laplace's method. Let us consider firstly the case $t>0$.

    Let $\phi(\theta):=\varepsilon\cos\theta-t \sin\theta$ with $\theta\in[-\pi,\pi]$. Then $\phi$ takes its local maximum and minimum at
    \begin{equation}
    \begin{split}
        &\theta_{\max}=\arctan(-t/\varepsilon)\in\left(-\frac{\pi}{2},0\right),\quad\theta_{\min}=\arctan(-t/\varepsilon)+\pi\in\left(\frac{\pi}{2},\pi\right).
    \end{split}
    \end{equation}
    Since $\theta_{\max}\subset (-\pi/2,0)$, the contribution of the integral is given at the right end point $\pi$. One can apply Laplace's method to find the approximation
    \begin{equation}
        \int_a^b \exp(R(\varepsilon \cos\theta-t \sin\theta))\dv \theta=O\left(\frac{e^{R\phi(\pi)}}{R|\phi'(\pi)|}\right).
    \end{equation}
    We also have
    \begin{equation}
        \phi(\pi)= \varepsilon,\quad \phi'(\pi)=-t\ne 0.
    \end{equation}
    Assembling the above gives the result.

    In the case $t<0$, the contribution of the integral is given at left end point $\pi$. The rest of the analysis remains the same.
\end{proof}

With the previous two lemmas, one can now proof Theorem~\ref{Fourier} fully rigorously.

\begin{proof}[Proof of Theorem~\ref{Fourier}]
    We firstly look at the case where $\rho$ has a compact support contained in an open subset of $[-M,M]$. Note that the support of $\rho^{(c)}$ is also contained in $[-M,M]$ by Theorem~\ref{conv_hull}.

    \vspace{1em}
    \smallskip\noindent\emph{Step I: Compactly supported $\rho$.}
    By Theorem~\ref{thm_tail}, the left side of~\eqref{fourier2} is finite. By Lemma~\ref{lem_Gamma} we have
    \begin{equation}
        \int_\R \rho^{(c)}(\dv x)e^{ ux}=\int_{-M}^M\rho^{(c)}(\dv x)\,{u\Gamma(c)}\frac{1}{2\pi i}\int_{\sgn(u)\mathcal C}\dv z\frac{e^{uz}}{(u(z-x))^c}.
    \end{equation}
    Because of the exponential decay of $e^{uz}$ along the contour $\sgn(u)\mathcal C$, one can apply a Fubini's theorem to interchange the integrals. Thus
    \begin{equation}
    \begin{split}
        \int_\R \rho^{(c)}(\dv x)e^{ux}&=y\Gamma(c)\frac{1}{2\pi i}\int_{\sgn(u)\mathcal C}\dv z\,e^{uz}\int_{-M}^M\frac{\rho^{(c)}(\dv x)}{(u(z-x))^c}\\
        &=u\Gamma(c)\frac{1}{2\pi i}\int_{\sgn(u)\mathcal C}\dv z\,\exp\left(uz-c\int_\R\log(u(z-s))\rho(\dv s)\right),
    \end{split}
    \end{equation}
    where we make use of the Markov-Krein correspondence between $\rho^{(c)}$ and $\rho$. Noting that
    \begin{equation}
        \log(u(z-s))=\log u +\log(z-s)+ 2m\pi i
    \end{equation}
    for some $m\in\Z$, and the $e^{2m\pi i}=1$, we can simplify the previous expressions to
    \begin{align}
        \int_\R \rho^{(c)}(\dv x)e^{ux}
        &=\frac{\Gamma(c)}{u^{c-1}}\frac{1}{2\pi i}\int_{\sgn(u)\mathcal C}\dv z\,\exp\left(uz-c\int_\R\log(z-s)\rho(\dv s)\right).
    \end{align}
    
    \vspace{1em}
    \smallskip\noindent\emph{Step II: $u=it$ with $t>0$.} In the special case where $\rho$ is compactly supported and $u=it$ with $t>0$, we need to deform the contour $\sgn(u)\mathcal C$ to $\R-i\delta$, for every $\delta>0$. For a graphical illustration of the deformation, see Fig.~\ref{contour2}.

    \begin{figure}[h]
\begin{center}
\begin{tikzpicture}[
    contour1/.style={thick, 
        postaction={decorate}, 
        decoration={markings, 
            mark=at position 0.15 with {\arrow{Stealth}},
            mark=at position 0.55 with {\arrow{Stealth}},
            mark=at position 0.88 with {\arrow{Stealth}}
        }
    },
    contour2/.style={thick, 
        postaction={decorate}, 
        decoration={markings, 
            mark=at position 0.5 with {\arrow{Stealth}}
        }
    },
    contour3/.style={thick, 
        postaction={decorate}, 
        decoration={markings, 
            mark=at position 0.5 with {\arrow{Stealth}}
        }
    },
    axis/.style={-{Stealth[scale=1.5]}, semithick},
]

\def\BigR{3.5}
\def\BigM{1}

\coordinate (O) at (0,0);
\coordinate (Mpos) at (\BigM, -\BigM);   
\coordinate (Mneg) at (-\BigM, -\BigM);  

\draw[contour1] 
    (-\BigM, \BigR-\BigM) node[above left] {\scriptsize $-M+i(R-M)$} 
    -- 
    (-\BigM, -\BigM) coordinate (bottom) node[below] {\scriptsize $-M-iM$} 
    -- 
    (\BigM, -\BigM) coordinate (top) node[below] {\scriptsize $M-i M$} 
    -- 
    (\BigM, \BigR-\BigM) node[above right] {\scriptsize $M+i(R-M)$};

\draw[contour2] 
    (-\BigR-\BigM, -\BigM) node[below] {\scriptsize $-M-R-iM$} 
    -- 
    (-\BigM, -\BigM) coordinate (bottom);

\draw[contour2]
    (\BigM, -\BigM) coordinate (top) 
    -- 
    (\BigR+\BigM, -\BigM) node[below] {\scriptsize $M+R-iM$};

\draw[contour2,blue] (Mpos) ++(\BigR, 0) arc (0:90:\BigR) 
    node[midway, above right] {$C_R^+$};

\draw[contour2,blue] (Mneg) ++(0, \BigR) arc (90:180:\BigR)
    node[midway, above left] {$C_R^-$};


\draw[axis] (-\BigR-\BigM-0.5, 0) -- (\BigR+\BigM+0.5, 0) node[above] {$\mathrm{Re}(z)$};
\draw[axis] (0, -\BigM-.5) -- (0, \BigR) node[left] {$\mathrm{Im}(z)$};
\node at (0,0) [below left] {\scriptsize $0$};
\end{tikzpicture}
\caption{Deformation of the contour from $-i\mathcal C$ to $\R-iM$.}
\label{contour2}
\end{center}
\end{figure}
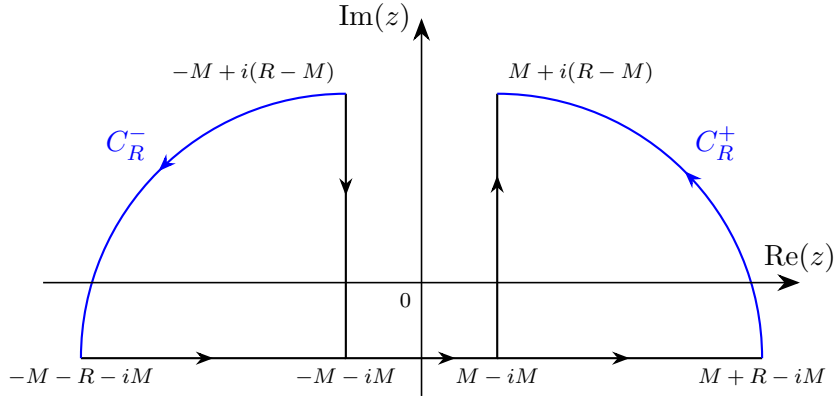
    
    Denote the arc going from $M+R-iM$ to $M+i(R-M)$ centering at $M-iM$ by $\mathcal C_R^+$, and the arc going from $-M+i(R-M)$ to $M+R-iM$ centering at $-M-iM$ by $\mathcal C_R^-$; they both have radius $R$. Then we have the inequality
    \begin{equation}
    \begin{split}
        \left|\int_{\mathcal C_R^+}\dv z\,e^{itz}\int_{-M}^M\frac{\rho^{(c)}(\dv x)}{(it(z-x))^c}\right|&=\left|\int_0^{\pi/2}Rie^{i\theta}\dv \theta\,e^{it(M-iM)}e^{itRe^{i\theta}}\int_{-M}^M\frac{\rho^{(c)}(\dv x)}{t^c(M-iM+Re^{i\theta}-x)^c}\right|\\
        &\leq \int_0^{\pi/2}\dv \theta\, Re^{tM-tR\sin\theta}\int_{-M}^M\frac{\rho^{(c)}(\dv x)}{t^c(R-3M)^c}=\frac{e^{tM}}{t^c(R-3M)^c}O(1)\to 0
    \end{split}
    \end{equation}
    as $R\to+\infty$, where we use the asymptotics
    \begin{equation}
        \int_0^{\pi/2}\dv\theta\, e^{-R\sin\theta}=O(R^{-1}),\text{ as }R\to+\infty,
    \end{equation}
    which can be deduced via a Laplace's method (c.f. Lemma~\ref{lem_laplace}). Similarly for $\mathcal C_R^-$ we have
    \begin{equation}
    \begin{split}
        \left|\int_{\mathcal C_R^-}\dv z\,e^{itz}\int_{-M}^M\frac{\rho^{(c)}(\dv x)}{(it(z-x))^c}\right|&=\left|\int_0^{\pi/2}Rie^{i\theta}\dv \theta\,e^{it(-M-iM)}e^{itRe^{i(\theta+\pi/2)}}\int_{-M}^M\frac{\rho^{(c)}(\dv x)}{t^c(-M-iM+Re^{i(\theta+\pi/2)}-x)^c}\right|\\
        &\leq \int_0^{\pi/2}\dv \theta\, Re^{tM-tR\cos\theta}\int_{-M}^M\frac{\rho^{(c)}(\dv x)}{t^c(R-3M)^c}=\frac{e^{tM}}{t^c(R-3M)^c}O(1)\to 0
    \end{split}
    \end{equation}
    as $R\to+\infty$. Therefore we now have deformed the contour to $\R-iM$.

    Denote the vertical contours $\gamma_{\pm R}$ as the one going from $\pm R-iM$ to $\pm R-i\delta$. Then
    \begin{equation}
    \begin{split}
        \left|\int_{\gamma_{\pm R}}\dv z\,e^{itz}\int_{-M}^M\frac{\rho^{(c)}(\dv x)}{(it(z-x))^c}\right|&=\left|\int_{M}^{\delta}(-i)\dv y\,e^{it(\pm R-iy)}\int_{-M}^M\frac{\rho^{(c)}(\dv x)}{t^c(\pm R-iy-x)^c}\right|\\
        &\le\int_M^{\delta}\dv y\,e^{ty}\int_{-M}^M\frac{\rho^{(c)}(\dv x)}{t^c(R-2M)^c}\to 0.
    \end{split}
    \end{equation}  
    Thus one has
    \begin{equation}
        \int_\R \rho^{(c)}(\dv x)e^{it x}=\frac{\Gamma(c)}{(it)^{c-1}}\frac{1}{2\pi i}\int_{\R-i\delta}\dv z\,\exp\left(itz-c\int_\R\log(it(z-s))\rho(\dv s)\right).
    \end{equation}
    which gives the final result for $t>0$.
    
    \vspace{1em}
    \smallskip\noindent\emph{Step III: $y=it$ with $t<0$.} Here on the contour $\sgn(y)\mathcal C$ the exponential function $e^{itz}$ decays. Thus similar contour deformations to the $t>0$ case can be applied to obtain the final contour $\R+i\delta$ with $\delta>0$. This finishes the proof in the case where $\rho$ is compactly supported.

    \vspace{1em}
    \smallskip\noindent\emph{Step IV: Non-compact suppport.}
    For the case where $\rho$ does not have a compact support, one makes use of the measure $\rho^{[M]}$ and the weak convergence $\rho^{[M]}\to\rho$ (by Lemma~\ref{weak_conv}). Noting that $\rho^{[M]}$ has a compact support, and thus by the Lebesgue dominated convergence theorem, we have the final result.

    \vspace{1em}
    \smallskip\noindent\emph{Step V: Special case.}
    In addition, for the case $u=it$, $t\ne 0$ and $c>1$,
    \begin{equation}
    \begin{split}
        \left|\int_{\R-i\delta\sgn(t)}\dv z\,e^{itz}\int_{\R}\frac{(\rho^{[M]})^{(c)}(\dv x)}{(z-x)^c}\right|&\leq e^{\delta \sgn(t)}\int_{\R}\dv y\, \int_{\R}\frac{(\rho^{[M]})^{(c)}(\dv x)}{((y-x)^2+\delta^2)^{c/2}}.
    \end{split}
    \end{equation}
    The integrand is absolutely integrable. Thus one can swap the $M$-limit with the $z$-integral to obtain
    \begin{equation}
        \int_\R \rho^{(c)}(\dv x)e^{it x}=\frac{\Gamma(c)}{(it)^{c-1}}\frac{1}{2\pi i}\int_{\R-i\delta\,\sgn(t)}\dv z\,e^{itz}\lim_{M\to+\infty}\int_{\R}\frac{(\rho^{[M]})^{(c)}(\dv x)}{(z-x)^c}.
    \end{equation}
    By Lemma~\ref{weak_conv}, the weak convergence $(\rho^{[M]})^{(c)}\to\rho^{(c)}$ finalises the proof.
\end{proof}




\begin{thebibliography}{9}

\bibitem{AG21}
Hardy, Adrien; Lambert, Gaultier. CLT for circular beta-ensembles at high temperature. J. Funct. Anal. 280 (2021), no. 7, Paper No. 108869, 40 pp.

\bibitem{AB19}
Akemann, Gernot; Byun, Sung-Soo. The high temperature crossover for general 2D Coulomb gases. J. Stat. Phys. 175 (2019), no. 6, 1043–1065.

\bibitem{ABG12}
Allez, R.; Bouchaud, J.-P.; Guionnet, A. Invariant Beta Ensembles and the Gauss-Wigner Crossover, Phys. Rev. Lett. (2012).

\bibitem{AGZ}
Anderson, Greg W.; Guionnet, Alice; Zeitouni, Ofer. An Introduction to Random Matrices. Cambridge Studies in Advanced Mathematics, vol. 118. Cambridge University Press, Cambridge, 2009.

\bibitem{ACG23}
Arizmendi, O.; Cébron G.; Gilliers N. Combinatorics of cyclic-conditional freeness. arXiv:2311.13178, pages
1–47, 2023.

\bibitem{BG_rectangular}
Benaych-Georges, Florent. Rectangular $R$-transform as the limit of rectangular spherical integrals. Journal of Mathematical Physics, 50(8):083302, 2009.

\bibitem{BGCG22}
Benaych-Georges, Florent; Cuenca, Cesar; Gorin, Vadim. Matrix addition and the Dunkl transform at high temperature. Comm. Math. Phys.394(2022), no.2, 735–795.

\bibitem{Benoit_thesis}
Collins, Benoit. Int\'egrales matricielles et Probabilit\'es Non-Commutatives. Mathématiques. Universit\'e Pierre et Marie Curie - Paris VI, 2003.

\bibitem{BG15}
Borodin, Alexei; Gorin, Vadim. General $\beta$-Jacobi corners process and the Gaussian free field. Comm. Pure Appl. Math. 68 (2015), no. 10, 1774–1844.

\bibitem{BG13}
Borot G., Guionnet A., Asymptotic expansion of $\beta$-matrix models, Comm. Math. Phys. (2013).

\bibitem{BEY14}
Bourgade, P.; Erd\"os, L.; Yau, H.-T. Universality of general $\beta$-ensembles. Duke Math. J. 163 (2014).

\bibitem{Bufetov13}
Bufetov, A. Kerov’s interlacing sequences and random matrices. J. Math. Phys., 54:113302, 10 pp, 2013.

\bibitem{CNXYZ22}
Chen, Peng; Nourdin, Ivan; Xu, Lihu; Yang, Xiaochuan; Zhang, Rui. Non-integrable stable approximation by Stein's method. J. Theoret. Probab. 35 (2022), no. 2, 1137–1186.

\bibitem{CRa}
Cifarelli, D.M.; Regazzini, E. (1979). A general approach to
Bayesian analysis of nonparametric problems. The associative mean values
within the framework of the Dirichlet process. I. (Italian) Riv. Mat.
Sci. Econom. Social. 2, 39–52. MR0573686

\bibitem{CRb}
Cifarelli, D.M.; Regazzini, E. (1979). A general approach to
Bayesian analysis of nonparametric problems. The associative mean values
within the framework of the Dirichlet process. II. (Italian) Riv. Mat.
Sci. Econom. Social. 2, 95–111.

\bibitem{CR90}
Cifarelli, D. M.; Regazzini, E. (1990). Distribution functions of means of a Dirichlet process, Ann. Statist., 18, 429–442.

\bibitem{Cuenca18}
Cuenca, Cesar. Asymptotic formulas for Macdonald polynomials and the boundary of the (q,t)-Gelfand-Tsetlin graph. SIGMA Symmetry Integrability Geom. Methods Appl. 14 (2018), Paper No. 001, 66 pp.

\bibitem{Cuenca18a}
Cuenca, Cesar. Pieri integral formula and asymptotics of Jack unitary characters. Selecta Math. (N.S.) 24 (2018), no. 3, 2737–2789.

\bibitem{Cuenca21}
Cuenca, Cesar. Universal behavior of the corners of orbital beta processes. Int. Math. Res. Not. IMRN 2021, no. 19, 14761–14813.

\bibitem{CD25}
Cuenca, C.; Dolega, M. Crystallization of discrete $N$-particle systems at high temperature. arXiv preprint arXiv:2510.23496, 2025.

\bibitem{DelloSchiavo19}
Dello Schiavo, Lorenzo. Characteristic functionals of Dirichlet measures. Electron. J. Probab.   24 (2019), Paper No. 115, 38 pp.

\bibitem{DGM24}
Dworaczek Guera, Charlie; Memin, Ronan. CLT for real $\beta$-ensembles at high temperature. Electron. J. Probab. 29 (2024), Paper No. 171, 45 pp.

\bibitem{dJ93}
de Jeu, M. FE. The dunkl transform. Invent. Math. 113, no. 1 (1993), pp. 147-162.

\bibitem{DK96}
Diaconis, P.; Kemperman, J. (1996). Some new tools on Dirichlet
priors. In Bayesian statistics 5 (J.M. Bernardo, J.O. Berger, A.P. Dawid
and A.F.M. Smith, eds.) 97–106. Oxford Univ. Press.

\bibitem{DE02}
Dumitriu, I.; Edelman A. Matrix models for $\beta$-ensembles. J. Math. Phys. 43 (2002).

\bibitem{Dunkl}
Dunkl, C. F. Differential-difference operators associated to reflection groups. Trans. Amer. Math. Soc. 311, no. 1 (1989), pp. 167-183.

\bibitem{Faraut}
Faraut, Jacques. Rayleigh theorem, projection of orbital measures and spline functions. Adv. Pure Appl. Math. 6 (2015), no. 4, 261–283.

\bibitem{FF16}
Faraut, Jacques; Fourati, Faiza. Markov-Krein transform. Colloq. Math. 144 (2016), no. 1, 137–156.

\bibitem{Ferguson73}
Ferguson, Thomas S. A Bayesian analysis of some nonparametric problems. Ann. Statist.1(1973), 209–230.

\bibitem{ForresterBook}
Forrester, P.J. Log-Gases and Random Matrices, LMS 34, Princeton University Press (2010).

\bibitem{Forrester22}
Forrester, Peter J. High-low temperature dualities for the classical $\beta$-ensembles. Random Matrices Theory Appl. 11 (2022), no. 4, Paper No. 2250035, 25 pp.

\bibitem{FM21}
Forrester, P. J.; Mazzuca, G. The classical $\beta$-ensembles with $\beta$ proportional to $1/N$: from loop equations to Dyson's disordered chain. J. Math. Phys. 62 (2021), no. 7, Paper No. 073505, 22 pp.

\bibitem{Fourati11a}
Fourati, F. (2011). Dirichlet distributions and orbital integrals, Journal of Lie theory, 21, 189–203.

\bibitem{Fourati11b}
Fourati, F. (2011). Distributions de Dirichlet, mesures orbitales et transformation de Markov-Krein. Th\`ese, Paris VI $\&$ Universit\'e ElManar, Tunis.

\bibitem{GK02}
Guhr, T.; Kohler, H. Recursive construction for a class of radial functions. I. Ordinary space. J. Math. Phys. 43, no. 5
(2002), pp. 2707-2740.

\bibitem{GN57}
Gelfand, I. M.; Neumark, M. A. Unit¨are Darstellungen der klassischen Gruppen. Akademie-Verlag, Berlin, 1957.

\bibitem{GM20}
Gorin, Vadim; Marcus, Adam W. Crystallization of random matrix orbits. Int. Math. Res. Not. IMRN 2020, no. 3, 883–913.

\bibitem{GuionnetHusson2021}
Guionnet, Alice; Husson, Jonathan.
Large deviations for top eigenvalues and a matrix extension of the spherical integral.
arXiv preprint, arXiv:2101.01983, 2021.

\bibitem{GM05}
Guionnet, A.; Maïda, M. A Fourier view on the R-transform and related asymptotics of spherical integrals. J. Funct. Anal.222(2005), no.2, 435–490.

\bibitem{GZ02}
Guionnet, Alice; Zeitouni, Ofer. Large deviations asymptotics for spherical integrals. J. Funct. Anal. 188 (2002), no. 2, 461–515.

\bibitem{HC}
Harish-Chandra. Differential operators on a semisimple Lie algebra. Amer. J. Math., 79:87–120, 1957.

\bibitem{HS94}
Heckman, G.; Schlichtkrull, H. Harmonic analysis and special functions on symmetric spaces, volume 16 of Perspectives in Mathematics. Academic Press, Inc., San Diego, CA, 1994

\bibitem{Helgason_book}
Helgason, Sigurdur. Groups and geometric analysis, volume 113 of Pure and Applied Mathematics. Academic Press, Inc., Orlando, FL, 1984. Integral geometry, invariant differential operators, and spherical functions.

\bibitem{Husson21}
Husson, Jonathan. Asymptotic behavior of multiplicative spherical integrals and S-transform. Random Matrices Theory Appl. 14 (2025), no. 3, Paper No. 2550008.

\bibitem{HussonKo2022}
Husson, Jonathan; Ko, Justin. Asymptotics of spherical integrals of sublinear rank. arXiv preprint, arXiv:2208.03642, 2022.

\bibitem{IZ}
Itzykson, C.; Zuber, J. B. The planar approximation. II. J. Math. Phys., 21(3):411–421, 1980.

\bibitem{Kerov_paper}
Kerov, S. Interlacing Measures. Amer. Math. Soc. Transl. (2), 181:35–83, 1998.

\bibitem{Kerov_book}
Kerov, S.V.; Tsilevich, N.S. 2003. Asymptotic representation theory of the symmetric group and its applications in analysis (Vol. 219, pp. xvi+-201). Providence, RI: American Mathematical Society.

\bibitem{KZ23}
Kieburg, Mario; Zhang, Jiyuan. Derivative principles for invariant ensembles. Adv. Math. 413 (2023), Paper No. 108833, 52 pp.

\bibitem{Krein_book}
Krein, M. G. On certain new studies in the perturbation theory for selfadjoint operators. In Topics in Differential and Integral Equations and Operator Theory, volume 7 of Operator Theory: Advances and Applications, pages 107–172. Birkhäser, 1983.

\bibitem{KM77}
Krein, M. G., Nudel’man, A. A. (1977). The Markov moment problem and extrema problems. Translations of Math. monographs, vol. 50, Amer. Math. Soc.

\bibitem{LP09}
Lijoi, Antonio; Prünster, Igor. Distributional properties of means of random probability measures.(English summary)
Stat. Surv. 3 (2009), 47–95.

\bibitem{LR04}
Lijoi, Antonio; Regazzini, Eugenio. Means of a Dirichlet process and multiple hypergeometric functions. Ann. Probab.32(2004), no.2, 1469–1495.

\bibitem{Macdonald_book}
Macdonald, I. G. Symmetric functions and Hall polynomials. Second edition. With contribution by A. V. Zelevinsky and a foreword by Richard Stanley. Reprint of the 2008 paperback edition Oxf. Class. Texts Phys. Sci. The Clarendon Press, Oxford University Press, New York, 2015. xii+475 pp.

\bibitem{Magaldi22}
Magaldi, Hugo. Study of beta-ensembles and their high-temperature limit. Statistics [math.ST]. Université Paris sciences et lettres, 2022. English.

\bibitem{Majumdar92}
Majumdar, Suman. On topological support of Dirichlet prior. Statist. Probab. Lett.   15 (1992), no. 5, 385–388.

\bibitem{Matveev19}
Matveev, K. Macdonald-positive specializations of the algebra of symmetric functions: Proof of the Kerov conjecture. Annals of Mathematics 189, no. 1 (2019), 277–316.

\bibitem{Mergny_thesis}
Mergny, Pierre. Spherical integrals and their applications to random matrix theory. Statistical Mechanics. Université Paris-Saclay, 2022. English.

\bibitem{MP20}
Mergny, Pierre; Potters, Marc. Asymptotic behavior of the multiplicative counterpart of the Harish-Chandra integral and the $S$-transform, arXiv:2007.09421 [math-ph], 2021.

\bibitem{MP22}
Mergny, Pierre; Potters, Marc. Rank one HCIZ at high temperature: interpolating between classical and free convolutions. SciPost Physics 12, no. 1 (2022): 022.

\bibitem{MT24}
Mingo, James A.; Tseng, Pei-Lun. Infinitesimal operators and the distribution of anticommutators and commutators. J. Funct. Anal. 287 (2024), no. 9, Paper No. 110591, 35 pp.

\bibitem{NT18}
Nakano, Fumihiko; Trinh, Khanh Duy. Gaussian beta ensembles at high temperature: eigenvalue fluctuations and bulk statistics. J. Stat. Phys. 173 (2018), no. 2, 295–321.

\bibitem{NT20}
Nakano, Fumihiko; Trinh, Khanh Duy. Poisson statistics for beta ensembles on the real line at high temperature. J. Stat. Phys. 179 (2020), no. 2, 632–649.

\bibitem{NTT25}
Nakano, F.; Trinh, H. D.; Trinh, K. D. Classical beta ensembles and related eigenvalues processes at high temperature and the Markov-Krein transform. J. Math. Phys. 66 (2025), no. 5, Paper No. 053304, 23 pp.

\bibitem{Opdam93}
Opdam, E. M. Dunkl operators, Bessel functions and the discriminant of a finite Coxeter group. Compos. Math. 85, no. 3 (1993), pp. 333-373.

\bibitem{Rosler03}
R\"osler, M. A positive radial product formula for the Dunkl kernel. Transactions of the American Mathematical Society 355, no. 6 (2003), 2413–2438.

\bibitem{ST81}
Sethuraman, Jayaram; Tiwari, Ram C. Convergence of Dirichlet measures and the interpretation of their parameter.Statistical decision theory and related topics, III, Vol. 2 (West Lafayette, Ind., 1981), 305–315. Academic Press, Inc. [Harcourt Brace Jovanovich, Publishers], New York-London, 1982

\bibitem{Sodin17}
Sodin, Sasha. Fluctuations of interlacing sequences. J. Math. Phys. Anal. Geom. 13 (2017), no. 4, 364–401.

\bibitem{Stanley89}
Stanley, R. P. Some Combinatorial Properties of Jack Symmetric Functions. Advances
in Mathematics 77 (1989), 76–115.

\bibitem{Sun16}
Sun, Y. A new integral formula for Heckman-Opdam hypergeometric functions. Adv. Math., 289:1157–1204, 2016

\bibitem{Yamamoto84}
Yamoto, H. (1984). Characteristic functions of means of distributions chosen from a Dirichlet process, Ann. Prob., 12, 262–267.

 
\end{thebibliography}
\end{document}